\documentclass[a4paper,twoside,11pt]{article}
\usepackage[margin=2.9cm]{geometry}
\usepackage[affil-it]{authblk}

\usepackage{amssymb,amsmath,amsthm,booktabs}
\usepackage{graphicx,color,hyperref,times} 
\usepackage{enumerate,paralist,xspace}
\usepackage[ruled,vlined,linesnumbered]{algorithm2e}
\usepackage{todonotes,subcaption}

\newcommand{\rephrase}[3]{\medskip\noindent\textbf{#1~#2.}\hspace{0.5ex}\emph{#3}}
\newcommand{\charge}{\ensuremath{\operatorname{charge}}}

\graphicspath{{figures/}}

\begin{document}
\title{Beyond-Planarity: Density Results for Bipartite Graphs}
\author[1]{Patrizio~Angelini}
\author[1]{Michael~A.~Bekos}
\author[1]{Michael~Kaufmann}
\author[1]{Maximilian~Pfister}
\author[2]{Torsten~Ueckerdt}
\affil[1]{Institut f{\"u}r Informatik, Universit{\"a}t T{\"u}bingen, Germany\\
\texttt{\small$\{$angelini,bekos,mk,pfister$\}$@informatik.uni-tuebingen.de}}
\affil[2]{Fakult\"at f\"ur Mathematik, KIT, Karlsruhe, Germany\\
\texttt{\small torsten.ueckerdt@kit.edu}}
\date{}

\newtheorem{lemma}{Lemma}
\newtheorem{theorem}{Theorem}
\newtheorem{definition}{Definition}
\newtheorem*{remark}{Remark}
\newtheorem{corollary}{Corollary}
\newtheorem{cl}{Claim}
\newtheorem{prp}{Property}

\maketitle

\begin{abstract}
\emph{Beyond-planarity} focuses on the study of geometric and topological graphs that are in some sense nearly-planar. Here, planarity is relaxed by allowing edge crossings, but only with respect to some local forbidden crossing configurations. Early research dates back to the 1960s (e.g., Avital and Hanani~\cite{avital-66}) for extremal problems on geometric graphs, but is also related to graph drawing problems where visual clutter by edge crossings should be minimized (e.g., Huang~et~al.~\cite{HuangHE08}) that could negatively affect the readability of the drawing. Different types of forbidden crossing configurations give rise to different families of nearly-planar graphs.

Most of the literature focuses on Tur\'an-type problems, which ask for the maximum number of edges a nearly-planar graph can have. Here, we study this problem for bipartite topological graphs, considering several types of nearly-planar graphs, i.e.~1-planar, 2-planar, fan-planar, and RAC graphs. We prove bounds on the number of edges that are tight up to small additive~constants; some of them are surprising and not along the lines of the known results for non-bipartite graphs. Our findings lead to an improvement of the leading constant of the well-known Crossing Lemma for bipartite graphs, as well as to a number of interesting research questions on topological graphs.   
\end{abstract}

\section{Introduction}
\label{sec:introduction}

Planarity has been one of the central concepts in the areas of graph algorithms, computational geometry, and graph theory since the beginning of the previous century. While planar graphs were originally defined in terms of their geometric representation (i.e., a topological graph is \emph{planar} if it contains no edge crossing), they exhibit a number of combinatorial properties that only depend on their abstract representation. To cite only some of the most important landmarks, we refer to the characterization of planar graphs in terms of forbidden minors, due to Kuratowski~\cite{k-spcgt-1930}, to the existence of linear-time algorithms to test graph planarity~\cite{bm-cespe-04,fmr-ttp-06,HopcroftT74}, to the Four-Color Theorem~\cite{appel1977-1,appel1977-2}, and to the Euler's polyhedron formula, which can be used to show that $n$-vertex planar graphs have at most $3n-6$ edges.

For the applicative purpose of visualizing real-world networks, however, the concept of planarity turns out to be  restrictive. In fact, graphs representing such networks are generally too dense to be planar, even though one can often confine non-planarity in some local structures. Recent cognitive experiments~\cite{HuangHE08} show that this does not affect too much the readability of the drawing, if these local structures satisfy specific properties. In other words, these experiments indicate that even non-planar drawings may be effective for human understanding, as long as the crossing configurations satisfy certain properties. Different requirements on the crossing configurations naturally give rise to different classes of topological or geometric, i.e. straight-line, \emph{nearly-planar} graphs. \emph{Beyond-planarity} is then defined as a generalization of planarity, which encompasses all these graph classes. Early works date back to the 1960's~\cite{avital-66} in the field of extremal graph theory, and continued over the years~\cite{AgarwalAPPS97,DBLP:journals/dcg/AlonE89,MANA:MANA3211170125,kupitz1979extremal,PachT97}; also due to the aforementioned experiments, a strong attention on the topic was recently raised~\cite{SoCG2017,Shonan2016,Dagstuhl2016}, which led to many results described below.

Some of the most studied nearly-planar graphs include:%
\begin {inparaenum}[(i)]
\item \emph{k-planar graphs}, in which each edge is crossed at most $k$ times~\cite{DBLP:journals/corr/Ackerman15,MANA:MANA3211170125,DBLP:journals/corr/KobourovLM17,PachRTT06,PachT97,MR0187232},~see Fig.~\ref{fig:1-planar};
\item \emph{k-quasiplanar graphs}, which disallow sets of $k$ pairwise crossing edges~\cite{AckermanT07,AgarwalAPPS97,FoxPS13}, see Fig.~\ref{fig:quasi};
\item \emph{fan-planar} graphs, in which no edge is crossed by two independent edges or by two adjacent edges from different directions ~\cite{BekosCGHK14,BinucciCDGKKMT17,BinucciGDMPST15,KaufmannU14},~see Fig.~\ref{fig:fan-planar};
\item \emph{RAC graphs}, in which edge crossings only happen at right angles~\cite{DidimoEL11,Didimo2013,EadesL13},~see Fig.~\ref{fig:rac}.
\end{inparaenum}
Two notable sub-families of $1$-planar graphs are the \emph{IC-planar graphs}~\cite{BrandenburgDEKL16,Zhang2013}, in which crossings are \emph{independent} (i.e., no two crossed edges share an endpoint), and the \emph{NIC-planar graphs}~\cite{Zhang2014}, in which crossings are \emph{nearly independent} (i.e., no two pairs of crossed edges share two endpoints). Other families include \emph{fan-crossing free graphs}~\cite{CheongHKK15}, \emph{planarly-connected graphs}~\cite{AckermanKV16}, and \emph{bar-k-visibility graphs}~\cite{DeanEGLST07,Evans0LMW14}.

\begin{figure}
	\centering
	\subcaptionbox{\label{fig:1-planar}1-planar}{\includegraphics[width=0.18\textwidth,page=1]{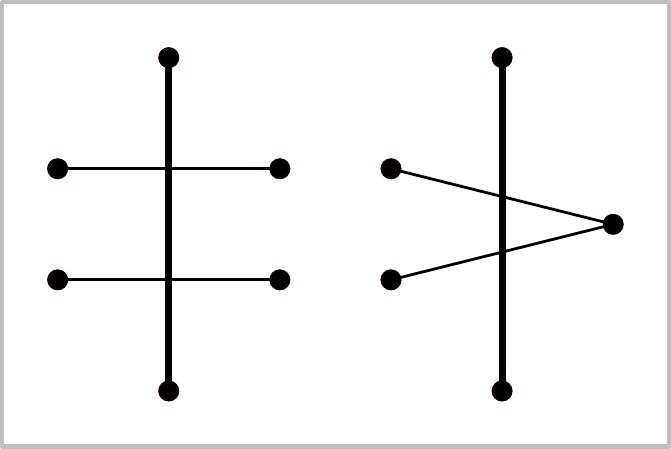}}
	\hfil
	\subcaptionbox{\label{fig:quasi}3-quasiplanar}{\includegraphics[width=0.18\textwidth,page=5]{beyond-planarity}}
	\hfil
	\subcaptionbox{\label{fig:fan-planar}fan-planar}{\includegraphics[width=0.18\textwidth,page=4]{beyond-planarity}}
	\hfil
	\subcaptionbox{\label{fig:rac}RAC}{\includegraphics[width=0.18\textwidth,page=3]{beyond-planarity}}
   \caption{Different forbidden crossing configurations.}
\label{fig:beyond-planarity}
\end{figure}

From the combinatorial point of view, the main question concerns the maximum number of edges for a graph in a certain class. This extreme graph theory question is usually referred to as a Tur\'an-type problem \cite{bollobas1986combinatorics}. Tight density bounds are known for several classes~\cite{DidimoEL11,KaufmannU14,PachRTT06,PachT97,Zhang2014,Zhang2013}; a main open question is to determine the density of $k$-quasiplanar graphs, which is conjectured to be linear in $n$ for any fixed~$k$~\cite{Ackerman09,AgarwalAPPS97,DBLP:journals/dcg/AlonE89,FoxPS13}.
Works on finding tight bounds on the edge density of $1$-,$2$-, $3$- and $4$-planar graphs have led to corresponding improvements on the leading constant of the well known Crossing Lemma~\cite{ajtal82,Leighton:1983:CIV:2304}; refer to~\cite{DBLP:journals/corr/Ackerman15,PachRTT06,PachT97}.
Another combinatorial question is to discover inclusion relationships between classes~\cite{AngeliniBBDD17,BinucciGDMPST15,BrandenburgDEKL16,DehkordiE12,EadesL13,Evans0LMW14,DBLP:journals/corr/HoffmannT17}.

From the complexity side, in contrast to efficient planarity testing algorithms~\cite{HopcroftT74}, recognizing a nearly-planar graph has often been proven to be NP-hard~\cite{ArgyriouBS12,BekosCGHK14,BinucciGDMPST15,BrandenburgDEKL16,GrigorievB07}. Polynomial-time testing algorithms can be found when posing additional restrictions on the produced drawings, namely, that the vertices are required to lie either on two parallel lines (\emph{2-layer} setting; see, e.g.,~\cite{BinucciCDGKKMT17,BinucciGDMPST15,DidimoEL10,GiacomoDEL14}) or on the outer face of the drawing (\emph{outer} setting; see, e.g.,~\cite{AuerBBGHNR16,BekosCGHK14,DehkordiEHN16,HongEKLSS15,HongN15}).

Each of these variants define new graph classes, which have also been studied in terms of their maximum density, e.g.~\cite{BekosCGHK14,BinucciCDGKKMT17,BinucciGDMPST15}.
Another natural restriction, which has been rarely explored in the literature, is to pose additional structural constraints on the graphs themselves, rather than on their drawings. For 3-connected 1-plane graphs, Alam et al.~\cite{AlamBK13} presented a polynomial-time algorithm to construct 1-planar straight-line drawings. Further, Brandenburg~\cite{Brandenburg16a} gave an efficient algorithm to recognize optimal 1-planar graphs, i.e., those with the maximum number of edges.

For the important class of bipartite graphs, very few results have been discovered yet. From the density point of view, the only result we are aware of is a tight bound of $3n-8$ edges for bipartite $1$-planar graphs~\cite{CzapPS16,Karpov2014}. Didimo et al.~\cite{DidimoEL10} characterize the complete bipartite graphs that admit RAC drawings, but their result does not extend to non-complete graphs.

\begin{table}[t!]
	\centering
	\caption{Summary of our results (from sparse to dense);
	the bound with an asterisk ($\ast$) is not tight.}
	\label{table:results}
	\begin{tabular}{rcc|cccc}
		\toprule
		& \multicolumn{2}{c}{General} & \multicolumn{4}{c}{Bipartite}\\
		\cmidrule{2-3} \cmidrule{4-7}
		Graph class~ & Bound (tight)~ & ~Ref.~ & ~Lower bound~ & ~Ref.~ & ~Upper bound~& ~Ref.~\\
		\midrule
		IC-planar:   & $3.5n-7$   & \cite{Zhang2013}   & $2.25n-4$ & Thm.\ref*{thm:ic-lower}  & $2.25n-4$   & Thm.\ref*{thm:ic-upper}\\
		NIC-planar:  & $3.6n-7.2$ & \cite{Zhang2014}   & $2.5n-5$  & Thm.\ref*{thm:nic-lower} & $2.5n-5$    & Thm.\ref*{thm:nic-upper}\\
		$1$-planar:  & $4n-8$     & \cite{MR0187232}   & $3n-8$    & \cite{CzapPS16}          & $3n-8$      & \cite{CzapPS16}\\
		RAC:         & $4n-10$    & \cite{DidimoEL11}  & $3n-9$    & Thm.\ref*{thm:rac-lower} & $3n-7$      & Thm.\ref*{thm:rac-upper}\\
		$2$-planar:  & $5n-10$    & \cite{PachT97}     & $3.5n-12$ & Thm.\ref*{thm:2-lower}   & $3.5n-7$    & Thm.\ref*{thm:2-upper}\\
		fan-planar:  & $5n-10$    & \cite{KaufmannU14} & $4n-16$   & Thm.\ref*{thm:fan-lower} & $4n-12$     & Thm.\ref*{thm:fan-upper}\\
		$3$-planar:  & $5.5n-11$  & \cite{DBLP:conf/gd/Bekos0R16}  & $4n-O(1)$ & Sec.\ref*{sec:conclusions} & ---  & --- \\
		$k$-planar:  & $3.81\sqrt{k}n^{~~\ast}$  & \cite{DBLP:journals/corr/Ackerman15}   & ---  & --- & $3.005 \sqrt{k}n$ & Thm.~\ref*{thm:general-bound} \\
		\bottomrule
	\end{tabular}
\end{table}

\subsection{Our contribution} 
Along this direction, we study in this paper several classes of nearly-planar bipartite topological graphs, focusing on Tur\'an-type problems. Table~\ref{table:results} shows our~findings. Note that the new bound on the edge density of bipartite $2$-planar graphs leads to an improvement of the leading constant of the Crossing Lemma for bipartite graphs from $\frac{1}{29} \approx 0.0345$~\cite{DBLP:journals/corr/Ackerman15} to $\frac{16}{289} \approx 0.0554$ (see Theorem~\ref{thm:crossing-lemma}), as well as to a new bound for the edge density of bipartite $k$-planar graphs (see Theorem~\ref{thm:general-bound}). 
Additionally, our results unveil an interesting, and somehow unexpected, tendency in the density of $k$-planar bipartite topological graphs with respect to the one of general $k$-planar graphs. At first sight, the differences seem to be around $n$, as it is in the planar and in the 1-planar cases (i.e., $n-2$). This turns out to be true also for RAC and fan-planar graphs. However, for the cases of IC- and NIC-planar graphs, and in particular for 2-planar graphs, the differences are surprisingly large. Considering ratios between the bounds instead of the differences, the results are even more~unexpected.

Another notable observation from our results is that, in the bipartite setting, fan-planar graphs can be denser than $2$-planar graphs, while in the non-bipartite case these two classes have the same maximum density, even though none of them is contained in the other~\cite{BinucciCDGKKMT17}. In Section~\ref{sec:conclusions} we discuss a number of open problems that are raised by our work.

\section{Methodology}
\label{sec:methodology}

We focus on five classes of bipartite nearly-planar graphs: IC-planar, NIC-planar, RAC, fan-planar and $2$-planar graphs; refer to Sections~\ref{sec:ic}--\ref{sec:twoplanar}. To estimate the maximum edge density of each class we employ different counting techniques.
\begin{enumerate}
\item For the class of bipartite IC-planar graphs, we apply a direct counting argument based on the number of crossings that are possible due to the restrictions posed by IC-planarity.
\item Our approach is different for the class of bipartite NIC-planar graphs. We show that~a bipartite NIC-planar graph of maximum density contains a set of uncrossed edges forming a plane subgraph whose faces have length $6$. The density bound is obtained by observing that one can embed exactly one crossing pair of edges inside each facial $6$-cycle.
\item To estimate the maximum number of edges of a bipartite RAC graph, we adjust a technique by Didimo et al.~\cite{DidimoEL11}, who proved the corresponding bound for general RAC~graphs.
\item For bipartite fan-planar graphs, our technique is more involved. We first examine structural properties when these graphs are maximal. Then, we show how to augment any of these graphs (by appropriately adding vertices and edges) such that it contains as a subgraph a planar quadrangulation. Based on this property, we develop a charging scheme which charges edges involved in fan crossings to the corresponding vertices, so that the difference between the degree of a vertex and its charge is at least 2. This implies that there are at least as many edges in the quadrangulation as in the rest of the graph.
\item Our approach for bipartite $2$-planar graphs follows similar lines. We show that maximal bipartite 2-planar graphs have a planar quadrangulation as a subgraph. We then use a counting scheme based on an auxiliary directed plane graph, defined by appropriately orienting the dual of the quadrangulation, which describes dependencies of adjacent quadrangular faces posed by the edges that do not belong to the quadrangulation.
\end{enumerate}

\section{Preliminaries}
\label{sec:preliminaries}
We consider connected \emph{topological} graphs, i.e., drawn in the plane with vertices represented by distinct points in $\mathbb R^2$ and edges by Jordan curves connecting their endvertices, so that: %
\begin{inparaenum}[(i)]
\item no edge passes through a vertex different from its endpoints,
\item no two adjacent edges cross,
\item no edge crosses itself,
\item no two edges meet tangentially, and
\item no two edges cross more than once.
\end{inparaenum}
A graph has no self-loops or multiedges. Otherwise, it is a \emph{multigraph}, for which we assume that the bounded and unbounded regions defined by self-loops or multiedges contain at least one vertex in their interiors, i.e., there are no \emph{homotopic}~edges. 

We refer to a nearly-planar graph $G$ with $n$ vertices and maximum possible number of edges as \emph{optimal}. Also, we denote by $G_p$ a maximal plane subgraph of $G$ on the same vertex-set as $G$, i.e., with the largest number of edges such that in the drawing of~$G_p$ inherited from~$G$ there exists no two edges crossing each other. We call $G_p$ the \emph{planar structure} of $G$. Let $f=\{u_0,u_1,\dots,u_{k-1}\}$ be a face of $G_p$. We say that $f$ is \emph{simple} if $u_i \neq u_j$ for each $i \neq j$, and it is \emph{connected} if edge $(u_i,u_{i+1})$ exists for each $i=0,\dots,k-1$ (indices modulo~$k$).

\section{Bipartite IC-planarity}
\label{sec:ic}
In this section, we give a tight bound on the density of bipartite IC-planar~graphs.

\begin{theorem}\label{thm:ic-lower}
There exist infinitely many bipartite $n$-vertex IC-planar graphs with $2.25n-4$ edges.
\end{theorem}
\begin{proof}
Fig.\ref{fig:max-ic} shows a construction that yields $n$-vertex bipartite IC-planar graphs with $2.25n-4$ edges. The graph is composed of a quadrangular grid of size $5 \times \frac{n}{4}$, which we wrap around a cylinder by identifying the vertices of its topmost row with the ones of its bottommost row. Thus, the two bases of the cylinder are faces of length $4$. Now, observe that each vertex in our construction participates in exactly one crossing. Hence, the number of skewed edges of our construction is exactly $n/4$. Since a planar graph with $n$ vertices, whose faces are of length~$4$, has exactly $2n-4$ edges, it follows that the constructed graph has  exactly $2n-4 + \frac{n}{4} = 2.25n - 4$~edges.
\end{proof}

\begin{theorem}\label{thm:ic-upper}
A bipartite $n$-vertex IC-planar graph has at most $2.25n-4$ edges
\end{theorem}
\begin{proof}
Our proof is an adjustment of the one for general IC-planar graphs~\cite{Zhang2013}. Let $G$~be a bipartite $n$-vertex optimal IC-planar graph. Let $cr(G)$ be the number of crossings of~$G$. Since every vertex of~$G$ is incident to at most one crossing, $cr(G) \leq \frac{n}{4}$. By removing one edge from every pair of crossing edges of~$G$, we obtain a plane bipartite graph, which has at most $2n-4$ edges. Hence, the number of edges of~$G$ is at most $2n-4 + cr(G) = 2.25n - 4$.
\end{proof}

\section{Bipartite NIC-planarity}
\label{sec:nic}
We continue our study with the class of bipartite NIC-planar graphs.

\begin{theorem}\label{thm:nic-lower}
There exist infinitely many bipartite $n$-vertex NIC-planar graphs with $2.5n-5$ edges.
\end{theorem}
\begin{proof}
Fig.\ref{fig:max-nic} shows a construction that yields $n$-vertex bipartite NIC-planar graphs with $2.5n-5$ edges. Again, the graph is composed of a quadrangular grid of size $5 \times \frac{n}{4}$, which we wrap around a cylinder by identifying the vertices of its topmost row with the ones of its bottommost row. Thus, the two bases of the cylinder are faces of length $4$. Now, observe that each vertex in our construction participates in exactly two crossings, except for exactly four vertices (marked by dotted circles) participating only in one crossing. Hence, the number of skewed edges of our construction is exactly $n/2-1$. Since a planar graph with $n$ vertices, whose faces are of length~$4$, has exactly $2n-4$ edges, it follows that the constructed graph has  exactly $2n-4 + \frac{n}{2} -1 = 2.5n - 5$~edges.
\end{proof}

\begin{figure}[t]
	\centering
	\subcaptionbox{\label{fig:max-ic}IC-planar}{\includegraphics[scale=0.48,page=1]{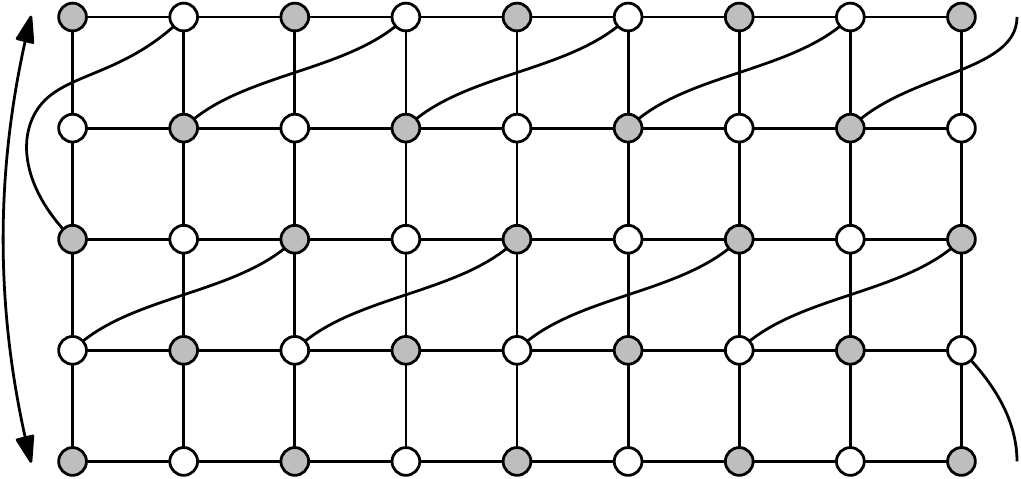}}
	\hfil
	\subcaptionbox{\label{fig:max-nic}NIC-planar}{\includegraphics[scale=0.48,page=4]{max-nic-graph}}
	\caption{
		Bipartite $n$-vertex IC- and NIC-planar graphs with
		(a)~$2.25n-4$ and
		(b)~$2.5n-5$~edges.}
	\label{fig:constructions}
\end{figure}

\begin{theorem}\label{thm:nic-upper}
A bipartite $n$-vertex NIC-planar graph has at most $2.5n-5$ edges.
\end{theorem}
\begin{proof}
Let $G$ be a bipartite optimal NIC-planar graph with $n$ vertices; among such graphs, we assume that $G$ is one with the maximum number of edges that are not involved in any crossing. Namely, $G$ is such that the plane (bipartite) subgraph $H$ obtained by removing every edge that is involved in a crossing in $G$ has maximum density. 

Next, we claim that each face of $H$ containing two crossing edges in $G$ is connected and has length $6$ (hence, every face of $H$ has length either $6$, if it contains two edges crossing in $G$, or $4$ otherwise due to bipartiteness and maximality). Consider any pair of edges $(u_1,u_3)$ and $(u_2,u_4)$ that cross in $G$; let $u_1$ and $u_4$ belong to the same partition of $G$, which implies that $u_2$ and $u_3$ belong to the other partition. By $1$-planarity and by the optimality of $G$, we can assume that edges $(u_1,u_2)$ and $(u_3,u_4)$ belong to $H$, and in particular that there exist copies of these edges in $H$ such that the two regions delimited by $(u_1,u_2)$, $(u_1,u_3)$, and $(u_2,u_4)$, and by $(u_3,u_4)$, $(u_1,u_3)$, and $(u_2,u_4)$, respectively, do not contain any vertex in their interior.
We now show that there exists a vertex $v$ and two edges $(u_1,v)$ and $(v,u_4)$ of $H$ such that the region delimited by $(u_1,v)$, $(v,u_4)$, $(u_1,u_3)$, and $(u_2,u_4)$ does not contain any vertex in its interior. Consider the edge $(u_1,v)$ such that edges $(u_1,u_2)$, $(u_1,u_3)$, and $(u_1,v)$ appear consecutively around $u_1$. If $(u_1,v)$ belongs to $H$, then we can assume that also $(v,u_4)$ belongs to $H$, due to the maximality of $G$, hence satisfying the required property. Otherwise, suppose for a contradiction that $(u_1,v)$ is crossed by some other edge $e$ in $G$; observe that $u_4$ is not an endpoint of $e$, due to NIC-planarity. We then remove $e$, hence making $(u_1,v)$ belong to $H$, and we add a copy of edge $(v,u_4)$ to $H$ so to satisfy the required property. Namely, we draw this edge as a curve that starts at $u_1$, follows $(v,u_1)$, then $(u_1,u_3)$, and finally $(u_2,u_4)$, till reaching $u_4$. Note that, if there exists another copy of edge $(v,u_4)$ in $G$, it did not cross edge $e$ before its removal; hence, this copy of $(v,u_4)$ is not homotopic with the new copy we added. Since this operation results in a graph $G'$ with the same number of edges of $G$, and in a graph $H'$ with more edges than $H$, we have a contradiction to the maximality of $H$.
Applying the same arguments we can prove that there exists a vertex $w$ such that edges $(u_2,w)$ and $(w,u_3)$ belong to $H$ and the region delimited by $(u_2,w)$, $(w,u_3)$, $(u_1,u_3)$, and $(u_2,u_4)$ does not contain any vertex in its interior. This concludes the proof of our claim.  

Let $\nu$ and $\mu$ be the number of vertices and edges of $H$, respectively. Clearly, $n=\nu$. Let also $\phi_4$ and $\phi_6$ be the number of faces of length $4$ and $6$ in $H$, respectively. We have that $2\phi_4+3\phi_6=\mu$. By Euler's formula, we also have that $\mu + 2 = \nu + \phi_4 + \phi_6$. Combining these two equations, we obtain: $\phi_4+2\phi_6=\nu-2$. So, in total the number of edges of $G$ is $\mu + 2 \phi_6 = 2\phi_4+5\phi_6 = 2n-4 + \phi_6$. By Euler's formula, the number of faces of length $6$ of a planar graph is at most $(n-2)/2$, which implies that $G$ has at most $2.5n-5$ edges.
\end{proof}

\section{Bipartite RAC Graphs}
\label{sec:rac}

In this section, we continue our study on bipartite beyond-planarity with the class of RAC graphs. We prove an upper bound on their density that is optimal up to a constant of $2$.

\begin{theorem}\label{thm:rac-lower}
There exist infinitely many bipartite $n$-vertex RAC graphs with $3n-9$ edges.
\end{theorem}
\begin{proof}
For any $k > 1$, we recursively define a graph $G_k$ by attaching six vertices and $18$ edges to $G_{k-1}$ as in Fig.~\ref{fig:max-rac-1}; the base graph $G_1$ is a hexagon containing two crossing edges (see Fig.~\ref{fig:max-rac-2}). So, $G_k$ has $6k$ vertices and $18k-10$ edges. Fig.~\ref{fig:max-rac-3} shows that $G_k$ is RAC: if $G_{k-1}$ has been drawn so that its outerface is a parallelogram (gray in Fig.~\ref{fig:max-rac-3}), then we can augment it to a RAC drawing of $G_k$ in which the outerface is a parallelogram whose sides are parallel to the ones of $G_{k-1}$. The bound is obtained by adding one more edge in the outerface of $G_k$ by slightly ``adjusting'' its drawing; see Fig.~\ref{fig:max-rac-4}.
\end{proof}

\begin{figure}[h]
	\centering
	\subcaptionbox{\label{fig:max-rac-1}}{\includegraphics[width=.2\textwidth,page=1]{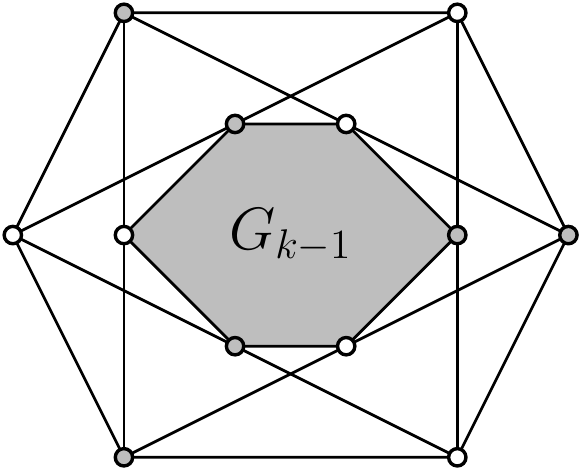}}
	\hfil
	\subcaptionbox{\label{fig:max-rac-3}}{\includegraphics[width=.7\textwidth,page=3]{max-rac}}
	\subcaptionbox{\label{fig:max-rac-2}}{\includegraphics[width=.2\textwidth,page=2]{max-rac}}
	\hfil
	\subcaptionbox{\label{fig:max-rac-4}}{\includegraphics[width=.7\textwidth,page=4]{max-rac}}
	\caption{Construction for a bipartite $n$-vertex RAC graph with $3n-9$ edges.}
	\label{fig:rac-construction}
\end{figure}

\begin{theorem}\label{thm:rac-upper}
A bipartite $n$-vertex RAC graph has at most $3n-7$ edges.
\end{theorem}

\begin{proof}
We use an argument similar to the one by Didimo et al.~\cite{DidimoEL11} to prove the upper~bound of $4n-10$ edges for general RAC graphs: Let $G$ be a (possibly non-bipartite) RAC graph with $n$ vertices. Since $G$ does not contain three mutually crossing edges, as in~\cite{DidimoEL11} we can color its edges with three colors (r, b, g) so that the crossing-free edges are the r-edges, while b-edges cross only g-edges, and vice-versa. Thus, the subgraphs~$G_{rb}$, consisting of only r- and b-edges, and $G_{rg}$, consisting of only r- and g-edges, are both planar. 
Didimo et al.~\cite[Lemma~4]{DidimoEL11} showed that each face of $G_{rb}$ has at least two r-edges, by observing that if this property did not hold, then the drawing could be augmented by adding r-edges. Thus, the number $m_b$ of b-edges is at most $n-1 - \lceil \lambda /2 \rceil$, where $\lambda \geq 3$ is the number of edges in the outer face of $G$.
Suppose now that~$G$ is additionally bipartite. We still have $m_b \leq n-1 - \lceil \lambda /2 \rceil$, but in this case $\lambda \geq 4$ holds (by bipartiteness). Hence, $m_b \leq n-3$. Since $G_{rg}$ is bipartite and planar, it has at most $2n-4$ edges (i.e., $m_r+m_g \leq 2n-4$). By combining the latter two inequalities, we obtain that the total number of edges of $G$ is at most $3n-7$, as desired.
\end{proof}

\section{Bipartite Fan-Planarity}
\label{sec:fanplanar}
We continue our study with the class of fan-planar graphs. We begin~as~usual with the lower bound (see Theorem~\ref{thm:fan-lower}), which we suspect to be best-possible both for graphs and multigraphs. For fan-planar bipartite graphs, we prove an almost tight upper bound (see Theorem~\ref{thm:fan-upper}).

\begin{theorem}\label{thm:fan-lower}
There exist infinitely many bipartite fan-planar %
\begin{inparaenum}[(i)]
\item graphs with $n$ vertices and exactly $4n-16$ edges, and
\item multigraphs with $n$ vertices and exactly $4n-12$ edges.
\end{inparaenum}
\end{theorem}
\begin{proof}
Recall that a fan-planar multigraph is a graph, possibly with multiedges, that admits a drawing with no homotopic edges in which for every edge $e$ all edges crossing $e$ have a common endpoint, which is moreover on the same side of $e$ (as we consider bipartite fan-planar graphs here, no loops occur). For the first part it suffices to consider $K_{4,n-4}$, $n\geq 5$, which has $4n-16$ edges and is known to be fan-planar for any~$n \geq 5$~\cite{KaufmannU14}. Another exceptional example is $K_{5,5}-e$, that is, $K_{5,5}$ minus an edge, which has $n=10$ vertices and $4n-16 = 24$ edges; see Fig.~\ref{fig:-K55-e}.

\begin{figure}[h]
	\centering
	\subcaptionbox{\label{fig:-K4n}}{\includegraphics[scale=1.2]{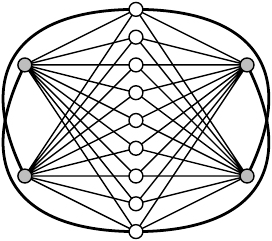}}
	\hfil
	\subcaptionbox{\label{fig:-K2n-multi}}{\includegraphics[scale=1.2]{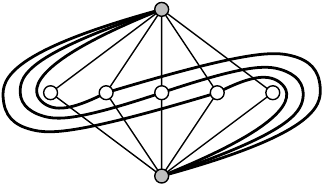}}
	\hfil
	\subcaptionbox{\label{fig:-K55-e}}{\includegraphics[scale=1.3]{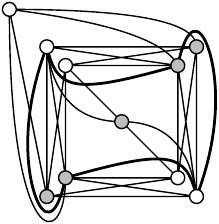}}
	\caption{Fan-planar drawings of bipartite multigraphs with $n$ vertices and $4n-12$ edges: (a)~$K_{4,n-4}$ with four additional multiedges (thick), (b)~$K_{2,n-2}$ with $2n-8$ additional multiedges (thick), (c)~$K_{5,5}-e$ with four additional multiedges (thick).}
	\label{fig:-fan-planar-bipartite}
\end{figure}

For the second part we observe that one can add four additional multiedges to the fan-planar drawings of $K_{4,n-4}$ and $K_{5,5}-e$ as illustrated in Figs.~\ref{fig:-K4n} and~\ref{fig:-K55-e}. Another class of examples is given by $K_{2,n-2}$, $n \geq 4$, to whose planar drawing one can add $2n-8$ additional multiedges in a fan-planar way as illustrated in Fig.~\ref{fig:-K2n-multi}, giving a fan-planar multigraph with $n$ vertices and $2(n-2)+2(n-4) = 4n-12$ edges. 
\end{proof}

\begin{remark}
Our upper bound (see Theorem~\ref{thm:fan-upper}) implies that the complete bipartite graphs $K_{5,9}$ and $K_{6,7}$ are not fan-planar. This is a big improvement over previous results, as the upper bound for general fan-planar graphs only implies that $K_{6,21}$ is not fan-planar, while it gives no guarantee for any $K_{5,n-5}$. However, we suspect that already $K_{5,5}$ is not fan-planar, which would follow from an upper bound of $4n-16$ edges for $n$-vertex fan-planar graphs and would resolve all remaining cases of complete bipartite graphs.
\end{remark}

To prove the upper bound, consider a bipartite fan-planar graph $G$ with a fixed fan-planar drawing. W.l.o.g.~assume that $G$ is edge-maximal and connected, and $A$, $B$ are the two bipartitions of $G$. Throughout this section we shall denote vertices in $A$ by $a$, $a'$, or $a_i$ for some index $i$, and similarly vertices in $B$ by $b$, $b'$, or $b_i$. By fan-planarity, for each edge $e$ of $G$ all edges crossing $e$ have a common endpoint $v$ (which also lies on the same side of $e$). We call $e$ an \emph{$A$-edge} (respectively, \emph{$B$-edge}) if this vertex $v$ lies in $A$ (respectively, $B$).

A \emph{cell} of some subgraph $H$ of $G$ is a connected component $c$ of the plane after removing all vertices and edges in $H$; see also~\cite{KaufmannU14}. The \emph{size} of $c$, denoted by $||c||$ is the total number of vertices and edge segments on the boundary $\partial c$ of $c$, counted with multiplicities.

\begin{lemma}[Kaufmann and Ueckerdt \cite{KaufmannU14}]\label{lem:fan-empty-cells}
Each fan-planar graph $G$ admits a fan-planar drawing such that if $c$ is a cell of any subgraph of $G$, and $||c|| = 4$, then $c$ contains no vertex of $G$ in its interior.
\end{lemma}

\noindent We choose a fan-planar drawing of $G$ with the property given in Lemma~\ref{lem:fan-empty-cells}.

\begin{corollary}\label{cor:A-edge-monotonicity}
If an edge $e = (a,b)$, with $a \in A$ and $b \in B$, is crossed in point $p$ by an $A$-edge $e'$, then every edge crossing $e$ between $a$ and $p$ is an $A$-edge that is moreover crossed by each edge that crosses $e'$.
\end{corollary}
\begin{proof}
Let $x$ be the common endpoint of all edges crossing $e$ and $e' = (x,y)$ be the $A$-edge crossing $e$ in $p$. Let $e'' = (x,y')$ be an edge that crosses $e$ between $p$ and $a$. If $e''$ is not an $A$-edge, it is crossed by an edge $e_1 = (a',b)$ with $a' \neq a$. The $A$-edge $e'$ is not crossed by $e_1$; see Fig.~\ref{fig:fan-planar-corollary-1}. But then there is a cell $c_1$ bounded by vertex $b$ and segments of $e$, $e''$ and $e_1$, which contains vertex $y$ in its interior, contradicting Lemma~\ref{lem:fan-empty-cells}. Symmetrically, if there is an edge $e_2 = (a,b')$ that crosses $e'$ but not $e''$ (see Fig.~\ref{fig:fan-planar-corollary-2}), then there is a cell $c_2$ bounded by vertex $a$ and segments of $e$, $e'$ and $e_2$, which contains vertex $y'$, again contradicting~Lemma~\ref{lem:fan-empty-cells}.
\end{proof}

\begin{figure}
	\centering
	\subcaptionbox{\label{fig:fan-planar-corollary-1}}{\includegraphics[page=1]{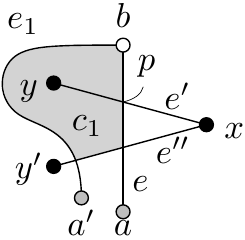}}
	\hfil
	\subcaptionbox{\label{fig:fan-planar-corollary-2}}{\includegraphics[page=2]{fan-planar-corollary}}
	\hfil
	\subcaptionbox{\label{fig:fan-planar-completion}}{\includegraphics{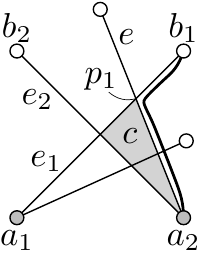}}
	\caption{Illustration of
		(a)-(b)~the proof of Corollary~\ref{cor:A-edge-monotonicity}, and
		(c)~Lemma~\ref{lem:fan-force-edges}.}
	\label{fig:fan-planar-proofs}
\end{figure}

\noindent Kaufmann and Ueckerdt~\cite{KaufmannU14} derive Lemma~\ref{lem:fan-empty-cells} from the following lemma.

\begin{lemma}[Kaufmann and Ueckerdt \cite{KaufmannU14}]\label{lem:fan-same-cell}
Let $G$ be given with a fan-planar drawing. If two edges $(v,w)$ and $(u,x)$ cross in a point $p$, no edge at $v$ crosses $(u,x)$ between $p$ and $u$, and no edge at $x$ crosses $(v,w)$ between $p$ and $w$, then $u$ and $w$ are contained in the same cell of $G$.
\end{lemma}

\noindent By the maximality of $G$ we have in this case that $(u,w)$ is an edge of $G$, provided $u$ and $w$ lie in distinct bipartition classes. We can use this fact to derive the following lemma.

\begin{lemma}\label{lem:fan-force-edges}
Let $e_1 = (a_1,b_1)$ and $e_2 = (a_2,b_2)$ be two crossing edges. If $e_1$ and $e_2$ are both $A$-edges or both $B$-edges, then $(a_2,b_1)$ is also contained in $G$ and can be drawn so that it crosses only edges that also cross $e_2$. If $e_1$ is an $A$-edge and $e_2$ is a $B$-edge, then $(a_2,b_1)$ is also contained in $G$ and can be drawn crossing-free.
\end{lemma}
\begin{proof}
First assume that $e_1$ and $e_2$ are both $A$-edges;
the case where $e_1$ and $e_2$ are both $B$-edges is analogous. Let $p_1$ be the crossing point on $e_1$ that is closest to $b_1$. Since $e_1$ is an $A$-edge crossing $(a_2,b_2)$, the edge $e$ crossing $e_1$ at $p_1$ (possibly $e = e_2$) is incident to $a_2$. Now either $e = e_2$ or the subgraph $H$ of $G$ consisting of $e$, $e_1$ and $e_2$ (and their vertices) has one bounded cell $c$ of size $4$, which by Lemma~\ref{lem:fan-empty-cells} contains no vertex of $G$. In both cases it follows that every edge of $G$ crossing $e$ between $a_2$ and $p_1$, also crosses $e_2$, and hence ends at $a_1$ (since $e_2$ is an $A$-edge crossing $(a_1,b_1)$).
We can conclude that drawing an edge from $b_1$ along $e_1$ to $p_1$ and then along $e$ to $a_2$ does not violate fan-planarity; see Fig.~\ref{fig:fan-planar-completion} for an illustration. Thus, by the maximality of $G$, edge $(a_2,b_1)$ is contained in $G$ .

Now assume that $e_1$ is an $A$-edge and $e_2$ is a $B$-edge. Let $p$ be the crossing point of $e_1$ and $e_2$. By Lemma~\ref{lem:fan-same-cell}, $a_2$ and $b_1$ lie on the same cell in $G$ and hence, by the maximality of $G$, we have that the edge $(a_2,b_1)$ is contained in $G$ and can be drawn crossing-free.
\end{proof}

\noindent We are now ready to prove the main theorem of this section. 

\begin{theorem}\label{thm:fan-upper}
Any $n$-vertex bipartite fan-planar graph has at most $4n - 12$ edges.
\end{theorem}
\begin{proof}
We start by considering the planar structure $G_p$ of $G$, i.e., an inclusion-maximal subgraph of $G$ whose drawing inherited from $G$ is crossing-free. Let $E_A$ and $E_B$ be the set of all $A$-edges and $B$-edges, respectively, in $E[G]-E[G_p]$. Each $e \in E_A$ is crossed by a non-empty (by maximality of $G_p$) set of edges in $G$ with common endpoint $a \in A$, and we say that $e$ \emph{charges} $a$. Similarly, every $e \in E_B$ charges a unique vertex~$b \in B$.

For any vertex $v$ in $G$, let $\charge(v)$ denote the number of edges in $E_A \cup E_B$ charging $v$. Moreover, for a multigraph $H$ containing $v$, let $\deg_H(v)$ denote the degree of $v$ in $H$, i.e., the number of edges of $H$ incident to $v$. Our goal is to show that for every vertex $v$ of $G$ we have $\deg_{G_p}(v) - \charge(v) \geq 2$.  However, this is not necessarily true when $G_p$ is not connected or has faces of length~$6$ or more. To overcome this issue, we shall add in a step-by-step procedure vertices and edges (possibly parallel but non-homotopic to existing edges in $G_p$) to the plane drawing of $G_p$ such that:
\begin{enumerate}[P.1]
\item the obtained multigraph $\bar{G_p}$ is a planar quadrangulation,\label{enum:still-planar}
\item the drawing of the multigraph $\bar{G} := G \cup \bar{G_p}$ is again fan-planar, and \label{enum:still-fan-planar}
\item each new vertex is added with three edges to other (possibly earlier added)~vertices.\label{enum:degree-3}
\end{enumerate}

To find $\bar{G_p}$, we first prove in the following claim that if $G_p$ is not a quadrangulation, we can add either one new edge or one new vertex with three new incident edges that do not cross any edge of $G_p$.
Moreover, the resulting multigraph (which may have parallel but non-homotopic edges) will still be bipartite and its drawing will still be fan-planar.

\begin{cl}\label{claim:fan-vertex-addition}
If $G_p$ is not a quadrangulation, one can add either one new edge or one new vertex with three new incident edges to the drawing of $G$, such that the resulting multigraph is still bipartite, the resulting drawing is still fan-planar, and the new edges do not cross any edge of $G_p$.
\end{cl}
\begin{proof}
First assume that $G_p$ is not connected. Then, there exists an edge $e = (a,b)$ in $G$ where $a \in A$ and $b \in B$ lie in different connected components of $G_p$ (w.l.o.g.\ $V[G_p] = V[G]$). As $e \notin E[G_p]$, there is an edge $e' = (a',b')$ in $G_p$ crossing $e$. By symmetry, we may assume that $a$ and $e'$ lie in different components of $G_p$.   Furthermore, w.l.o.g.\ $e'$ is the edge of $G_p$ whose crossing $p$ with $e$ is closest to $a$. We distinguish four cases.

\begin{description}
\item[Case~1. $e$ and $e'$ are $A$-edges.]
Then by Lemma~\ref{lem:fan-force-edges} there is an edge $(a,b')$ in $G$ and it can be drawn so that it crosses only edges that cross $e$ between $p$ and $a$; see Fig.\ref{fig:-fan-planar-connectivity-1}. None of the latter edges are in $G_p$ by the choice of $e'$. Hence $e$ can be added to $G_p$, contradicting the maximality of $G_p$ and that $a$ and $e'$ are in different components of $G_p$.

\item[Case~2. $e$ and $e'$ are $B$-edges.]
Again by Lemma~\ref{lem:fan-force-edges} there is an edge $(a,b')$ in $G$ and this time it can be drawn so that it crosses only edges that cross $e'$ between $p$ and $b'$; see Fig.\ref{fig:-fan-planar-connectivity-2}. None of the latter edges are in $G_p$ as they cross $e'$ which is in $G_p$.    Hence $e$ can like in Case~1 be added to $G_p$, arriving at the same contradiction.

\begin{figure}
	\centering
	\subcaptionbox{\label{fig:-fan-planar-connectivity-1}}{\includegraphics{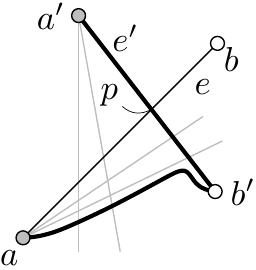}}
	\hfil
	\subcaptionbox{\label{fig:-fan-planar-connectivity-2}}{\includegraphics{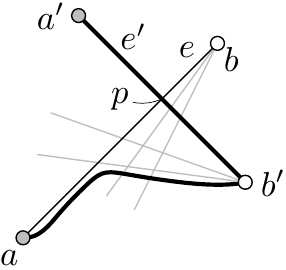}}
	\hfil
	\subcaptionbox{\label{fig:-fan-planar-connectivity-3}}{\includegraphics{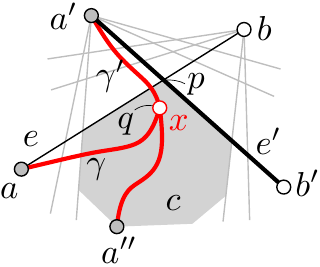}}
	\hfil
	\subcaptionbox{\label{fig:-fan-planar-connectivity-4}}{\includegraphics{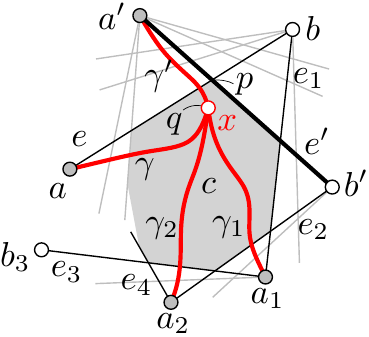}}
	\caption{Illustrations for the proof of Theorem~\ref{thm:fan-upper}. 
	Edges in $G_p$ are drawn thick, newly added vertices and edges are drawn in red.
	In~(\subref{fig:-fan-planar-connectivity-4}) an edge of the form $(a_i,x)$, $i \in \{1,2\}$, is added only if the edge $e_{i+1}$ is not in $G_p$.}
	\label{fig:-fan-planar-connectivity}
\end{figure}

\item[Case~3. $e$ is a $B$-edge and $e'$ is an $A$-edge.]
Here Lemma~\ref{lem:fan-force-edges} immediately gives that there is an edge $(a,b')$ in $G$ that can be drawn without crossings. Hence, as in the cases before, $(a,b')$ is in $G_p$, putting $a$ and $e'$ in the same component of $G_p$.

\item[Case~4. $e$ is an $A$-edge and $e'$ is a $B$-edge.]
This case is more elaborate. Consider a point $q$ in the plane very close to $p$ and on the $B$-side of $e$ and the $A$-side of $e'$; see Fig.\ref{fig:-fan-planar-connectivity-3}.

First, we claim that every edge $\tilde{e}$ crossing $e$ between $p$ and $a$ is an $A$-edge. In fact, as $\tilde{e} \notin E[G_p]$ (by choice of $e'$), it is crossed by some edge in $G_p$, but if $\tilde{e}$ were a $B$-edge, then by Corollary~\ref{cor:A-edge-monotonicity} this edge in $G_p$ would also cross $e' \in E[G_p]$, which is impossible. Hence we can draw a curve $\gamma$ from $a$ to $q$ crossing only $A$-edges that also cross $e$. Thus, $\gamma$ does not cross any edge of $G_p$.

Second, we claim that every edge $\tilde{e}$ crossing $e'$ between $p$ and $a'$ is an $A$-edge. In fact, this follows from Corollary~\ref{cor:A-edge-monotonicity} and the fact that $e$ is an $A$-edge. Hence we can draw a curve $\gamma'$ from $a'$ to $q$ crossing only $A$-edges that also cross $e'$. Note that $\gamma'$ does not cross any edge of $G_p$.

Now consider the cell $c$ of $G$ containing point $q$. If the boundary $\partial c$ of $c$ contains some vertex $b''$ from $B$, we can extend $\gamma$ and $\gamma'$ to two edges $(a,b'')$ and $(a',b'')$ respectively without creating any further crossings. Note that this drawing is again fan-planar. These two edges do not cross any edge of $G_p$, and as $a$ and $e'$ are in different components of $G_p$, at least one such edge is not already present in $G_p$ and we are done.
    
If the boundary $\partial c$ of $c$ contains some vertex $a'' \neq a,a'$ from $A$, we can add a new vertex $x$ to $B$ at point $q$ and draw edges $(a,x)$, $(a',x)$, and $(a'',x)$; see Fig.\ref{fig:-fan-planar-connectivity-3}. The resulting drawing is still fan-planar and new edges do not cross any edge of $G_p$, as~desired.

Finally, we assume that $\partial c$ contains, expect for possibly $a$, no vertex of $G$. Let us start~tracing $\partial c$ beginning with $p$ and following $e'$ towards $b'$. At some point we encounter a crossing of $e'$ with another edge $e_1$. Since $e'$ is a $B$-edge, $e_1 = (a_1,b)$ for some $a_1 \neq a,a'$. We follow $\partial c$ along $e_1$ towards $a_1$ and encounter a crossing of $e_1$ with another edge $e_2$; see Fig.\ref{fig:-fan-planar-connectivity-4}. If $e_2$ would be incident to $a'$, then by Corollary~\ref{cor:A-edge-monotonicity} $e_2$ would cross $e$ between $p$ and $a$ and hence would be a $B$-edge. Then, again by Corollary~\ref{cor:A-edge-monotonicity}, every edge crossing $e_2$ would also cross $e'$, which is in $E[G_p]$, which gives that $e_2 \in E[G_p]$, contradicting the choice of $e'$. Thus, $e_2$ is incident not to $a'$ but to $b'$, making $e_1$ a $B$-edge. Let $a_2$ denote the other endpoint of $e_2$, $a_2 \neq a_1,a'$, possibly $a_2 = a$. If $e_2$ is a $B$-edge, then every edge crossing $e_2$ also crosses $e'$ (Corollary~\ref{cor:A-edge-monotonicity}) and hence $e' \in E[G_p]$. Moreover, $e_2$ is not crossed between $a_2$ and its crossing with $e_1$ as $c$ is a cell. So $a_2$ lies on $\partial c$ and we have $a_2 = a$, which with $e_2 = (a_2,b') \in E[G_p]$ contradicts that $a$ and $e'$ are in different components of $G_p$.

Thus $e_2$ is an $A$-edge and we can draw a curve $\gamma_1$ from $a_1$ to $c$ crossing only $A$-edges that also cross $e_1$. We continue to follow $\partial c$ along $e_2$ towards $a_2$. As $e_2$ is not crossed between $b'$ and $e_1$ (by Lemma~\ref{lem:fan-empty-cells} and the fact that $e'$ is a $B$-edge), we encounter a crossing of $e_2$ with another edge $e_3 = (a_1,b_3)$ for some $b_3 \neq b,b'$. Following $\partial c$ along $e_3$ towards $b_3$ we encounter another crossing (as $b_3 \notin \partial c$) with some edge $e_4$. This edge $e_4$ is not incident to $b'$ (Lemma~\ref{lem:fan-empty-cells} and the fact that $c$ is a cell) and thus $e_4$ is incident to $a_2$, making $e_3$ an $A$-edge. So we can draw a fourth curve $\gamma_2$ from $a_2$ to $c$ crossing only $A$-edges that also cross $e_2$. Finally, by Lemma~\ref{lem:fan-empty-cells} every edge crossing $\gamma_1$ crosses every edge that crosses $\gamma_2$. Thus, there is $i \in \{1,2\}$ such that $\gamma_i$ does not cross any edge of $G_p$. We can now introduce a new vertex $x$ to $B$ into cell $c$ with edges $(a,x)$ along $\gamma$, $(a',x)$ along $\gamma'$, and $(a_i,x)$ along $\gamma_i$; see Fig.\ref{fig:-fan-planar-connectivity-4}. The resulting drawing is still fan-planar and the new edges do not cross any edge of $G_p$, as desired.
\end{description}

So from now on we may assume that $G_p$ is connected with $V[G_p] = V[G]$. If $G_p$ is not a quadrangulation, then there exists a face $f$ whose facial walk $W$ has length at least~$6$. For each edge $e \in E_A \cup E_B$ that intersects $f$ we have that $e \cap f$ consists of one or more \emph{segments}, where for each segment either both ends are crossing points on edges of $G_p$, or one end is such a crossing point and the other end is a vertex of $W$, called a \emph{stick}. If a stick $s$ has an end at vertex $a \in V[W] \cap A$, its other end is crossing an edge $(a',b') \in E[G_p]$. We call the part of $f - s$ containing $a'$ the \emph{outer} side of $s$ and the part of $f-s$ containing $b'$ the \emph{inner} side of $s$. The inner and outer sides of sticks with an end in $V[W] \cap B$ are defined analogously. A stick is call \emph{short} if its inner side contains only two vertices (one being the stick's end) and \emph{long}, otherwise. We distinguish two cases.
  
\begin{description}
\item[\boldmath Case~1. There is a long stick $s$.]
W.l.o.g.\ let $a \in V[W] \cap A$ be one end of $s$, and let $e' = (a',b')$ be the edge of $G_p$ containing the other end $p$ of $s$, where $a'$ lies on the outer side and $b'$ lies on the inner side of $s$. Let $e = (a,b)$ be the edge of $G$ corresponding to $s$, and assume w.l.o.g.\ that no edge incident to $a$ crosses $e'$ between $p$ and $b'$; see Fig.\ref{fig:-fan-planar-quadrangulation-1}. Now we are in similar situation as in the case of a disconnected $G_p$ above and we can argue along the same~lines.
    
First, if $e'$ is an $A$-edge or $e$ is a $B$-edge, then $(a,b')$ is an edge of $G$ by Lemma~\ref{lem:fan-force-edges} that can be drawn not crossing any edge in $G_p$. Hence, $(a,b') \in E[G_p]$ by maximality of $G_p$ and we can draw a parallel copy of $(a,b')$ in the specified way, and we are done. It remains the (more elaborate) case that $e'$ is a $B$-edge and $e$ is an $A$-edge. As above, consider a point $q$ in the plane very close to $p$ and on the $B$-side of $e$ and the $A$-side of $e'$; see Fig.\ref{fig:-fan-planar-quadrangulation-1}. By fan-planarity, every edge crossing $s$ (i.e., crossing $e$ between $p$ and $a$) corresponds to a long stick with one end being $a'$ that crosses the edge $e''$ at $a$ on the inner side of $s$. In particular, each edge crossing $s$ is an $A$-edge, and we can draw a curve $\gamma$ from $a$ to $q$ crossing only $A$-edges that also cross $e$. Similarly, every edge crossing $e'$ between $p$ and $a'$ is incident to $b$ and by Corollary~\ref{cor:A-edge-monotonicity} an $A$-edge just like $e$. Hence we can draw a curve $\gamma'$ from $a'$ to $q$ crossing only $A$-edges that also cross $e'$.
 
Now consider the cell $c$ of $G$ containing $q$. Note that $c$ is completely contained in face $f$ of $G_p$. Following the same argumentation as above, there is either a vertex $b''$ from $B$ on $\partial c$ and we can extend $\gamma'$ to an edge $(a',b'')$ without further crossings, or there is a vertex $a'' \neq a,a'$ and we can add a new vertex $x$ at position $q$, draw edge $(a,x)$ using $\gamma$, edge $(a',x)$ using $\gamma'$, and edge $(a'',x)$ in a fan-planar way not crossing any edge in $G_p$; see Fig.~\ref{fig:-fan-planar-quadrangulation-2} for one possible scenario. In both cases we are done.

\begin{figure}
	\centering
	\subcaptionbox{\label{fig:-fan-planar-quadrangulation-1}}{\includegraphics{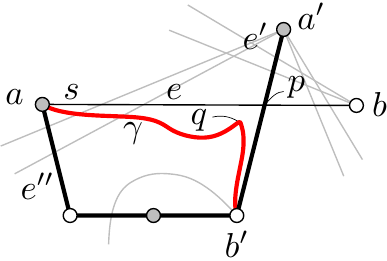}}
	\hfil
	\subcaptionbox{\label{fig:-fan-planar-quadrangulation-2}}{\includegraphics{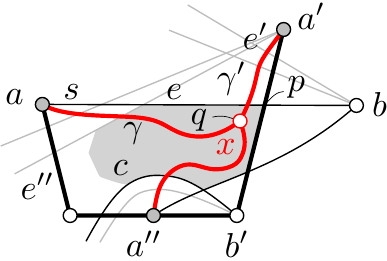}}
	\hfil
	\subcaptionbox{\label{fig:-fan-planar-quadrangulation-3}}{\includegraphics{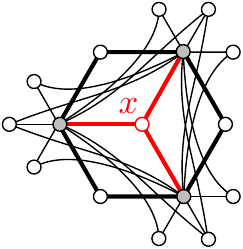}}
	\caption{Illustrations for the proof of Theorem~\ref{thm:fan-upper}. 
	Edges in $H_p$ are drawn thick, newly added vertices and edges are drawn in red.}
	\label{fig:-fan-planar-quadrangulation}
\end{figure}

\item[Case~2. All sticks are short.]
Consider a vertex $a \in V[W] \cap A$ that is on the inner side of a stick $s$ with end in $B$. If the edge corresponding to $s$ is a $B$-edge, we say that vertex $a$ is \emph{blocked}. If all such sticks correspond to $A$-edges, we say that vertex $a$ is \emph{semi-free}. And if there is no such stick for $a$, we say that $a$ is \emph{free} (again, blocked, semi-free, and free are defined analogously for vertices in $V[W] \cap B$).

The crucial observation is that if a vertex is blocked, then its two neighbors on $W$ are free. Hence, there exists a free vertex $a\in A \cap V[W]$ and a free vertex $b \in B \cap V[W]$ and $(a,b)$ can be added crossing-free to $G_p$, or three consecutive vertices in $A \cap V[W]$ (or in $B \cap V[W]$) are not blocked, in which case we can add a new vertex of degree~$3$ to $G_p$; see Fig.~\ref{fig:-fan-planar-quadrangulation-3}.
\end{description}

\noindent So in all cases, we can add a new edge or a new vertex with three new incident edges to $G$, such that the resulting drawing is still fan-planar, the resulting graph is still bipartite, and new edges do not cross any edge of $G_p$.
\end{proof}

Adding to $G_p$ an edge or a vertex with three edges, strictly increases the average degree in $G_p$. Hence, we ultimatively obtain supergraphs $\bar{G}$ of $G$ and $\bar{G_p}$ of $G_p$ satisfying~P.\ref{enum:still-planar}--P.\ref{enum:degree-3}. Next, we show that the charge of every original vertex $v$ is at most its degree in $\bar{G_p}$ minus~$2$.

\begin{cl}\label{claim:fan-charge}
Every $v \in V[G]$ satisfies $\deg_{\bar{G_p}}(v) - \charge(v) \geq 2$.
\end{cl}
\begin{proof}
Without loss of generality consider any $a \in A$ and let $k := \deg_{\bar{G_p}}(a)$ and $S \subseteq E_A$ be the set of edges charging $a$. First observe that no two edges of $S$ can cross. In fact, if $(a_1,b_1) \in E_A$ charges~$a$ and $(a_2,b_2) \in E_A$ crosses $(a_1,b_1)$, then $(a_2,b_2)$ charges $a_1 \neq a$.
Now consider the face $f$ of~$\bar{G_p}$ created by removing $a$ from the graph, and the closed facial walk $W$ around $f$. Walk~$W$ has length exactly $2k$ (counting with repetitions) as $\bar{G_p}$ is a quadrangulation. Moreover, every edge in $S$ lies completely in $f$ and has both endpoints on $W$. Hence, the subgraph of $\bar{G}$ consisting of all edges in $W \cup S$ is crossing-free and has vertex set $V[W]$. Define a new graph $J$ by breaking the repetitions along the walk $W$, i.e., $J$ consists of a cycle of length $2k$ and every edge in $S$ is an uncrossed chord of this cycle. As $J$ is outerplanar and bipartite, it has at most $k-2$ chords. Thus, $|S| = \charge(a) \leq k-2 = \deg_{\bar{G_p}}(a)-2$, as desired.
\end{proof}
\noindent Let $X = V[\bar{G}] - V[G]$ be the set of newly added vertices. For each $x \in X$, we have $\deg_{\bar{G_p}}(x) \geq 3$ and $\charge(x) = 0$. Thus, $\deg_{\bar{G_p}}(x) - \charge(x) \geq 3$, and by Claim~\ref{claim:fan-charge} we get
\[
2|E[\bar{G_p}]| - (|E_A| + |E_B|) = \sum_{v \in V[\bar{G_p}]} \left(\deg_{\bar{G_p}}(v) - \charge(v)\right) \geq 2n + 3|X| 
\]
\[
\Rightarrow |E_A| + |E_B| \leq 2|E[\bar{G_p}]| - 2n - 3|X|.
\]
On the other hand, $|E[G_p]| + 3|X| \leq |E[\bar{G_p}]|$ by~P.\ref{enum:degree-3} and $|E[\bar{G_p}]| = 2(n+|X|) - 4$ by~P.\ref{enum:still-planar}, which together give
\[
|E[G]| = |E[G_p]| + |E_A| + |E_B| \leq 3|E[\bar{G_p}]| - 6|X| - 2n = 4n-12
\]
and conclude the proof.
\end{proof}

\section{Bipartite 2-planarity}
\label{sec:twoplanar}

In this section we give an almost tight bound on the density of bipartite $2$-planar graphs. We start as usual with the lower bound.

\begin{theorem}\label{thm:2-lower}
There exist infinitely many bipartite $n$-vertex $2$-planar %
\begin{inparaenum}[(i)]
 \item graphs with exactly $3.5n-12$ edges, and
 \item multigraphs with exactly $3.5n-8$ edges.
\end{inparaenum}
\end{theorem}
\begin{proof}
For the first part, Fig.\ref{fig:max-2-planar-1} shows a construction that yields $n$-vertex bipartite $2$-planar graphs with $3.5n-12$ edges. As in the proofs of Theorems~\ref{thm:ic-lower} and~\ref{thm:nic-lower},  the graph is composed of quadrangular grid of size $5 \times \frac{n}{4}$, which we wrap around a cylinder by identifying the vertices of its topmost row with the ones of its bottommost row. Thus, the two bases of the cylinder are faces of length $4$. For each triple of consecutive faces on the same row, we can draw three edges violating neither bipartiteness nor $2$-planarity, as in the gray shaded area in Fig.\ref{fig:max-2-planar-1}. Observe that consecutive triples in the same row share one face. This implies that on average, one can draw three edges for every two faces (except for the four faces adjacent to the outermost quadrangles). Finally, we can add two additional edges inside each of the innermost and the outermost quadrangle (dashed in Fig.\ref{fig:max-2-planar-1}). Hence, the constructed graph has $n$ vertices and in total $2n-4 + 3 \cdot \frac{1}{2} \cdot (n-2-4-2) + 4 = 3.5n - 12$ edges. For the second part, observe that if we allow non-homotopic multiedges, then we can add four additional edges; see Fig.~\ref{fig:max-2-planar-2}
\end{proof}

\begin{figure}[h!]
	\centering
	\subcaptionbox{\label{fig:max-2-planar-1}}{\includegraphics[scale=0.5,page=1]{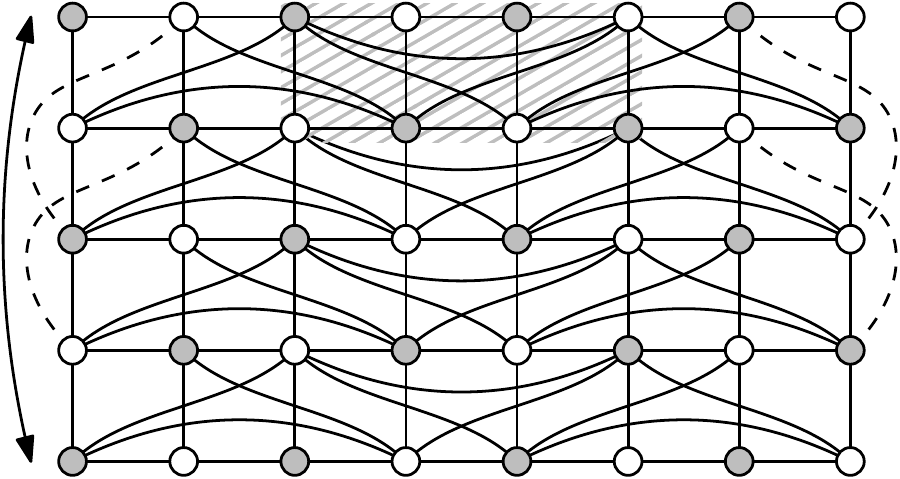}}
	\hfil
	\subcaptionbox{\label{fig:max-2-planar-2}}{\includegraphics[scale=0.5,page=2]{max-2-planar-graph}}
	\caption{Constructions for a bipartite $n$-vertex:
		(a)~$2$-planar graph with $3.5n-12$ edges, and
		(b)~$2$-planar multigraph with $3.5n-8$ edges.}
	\label{fig:max-2-planar}
\end{figure}

Since the proof of the upper bound is quite technical, we first give a high level description of the main steps of this proof (see Section \ref{subsec:twoplanar-overview}). The details of the proof are then given later in this section (see Section \ref{subsec:twoplanar-details}).

\subsection{The overview of our approach}
\label{subsec:twoplanar-overview}

In our proof, we first study structural properties of the planar structure $G_p$ of an optimal bipartite $2$-planar graph $G$. Let $(u,v)$ be an edge of $G$ that does not belong to~$G_p$. By the maximality of $G_p$, edge $(u,v)$ has at least one crossing with an edge of $G_p$. The part of $(u,v)$ that starts from $u$ (from $v$) and ends at the first intersection point of $(u,v)$ with an edge of $G_p$ is a \emph{stick} of $u$ (of $v$). When $(u,v)$ has exactly two crossings, there is a part of it that is not a stick, which we call \emph{middle-part}. Each part of an edge, either stick or middle-part, lies inside a face $f$ of $G_p$. In this case, we say that $f$ \emph{contains} this part. Let $f=\{u_0,u_1,\ldots,u_{k-1}\}$ be a face of $G_p$ with $k \geq 4$ and let $s$ be a stick of $u_i$, for some $i \in \{0,1,\ldots,k-1\}$, contained in $f$. We call~$s$ a \emph{short} stick, if it ends either at $(u_{i+1},u_{i+2})$ or at $(u_{i-1},u_{i-2})$ of $f$; otherwise, $s$ is called a \emph{long} stick; see Figs.~\ref{fig:short}-\ref{fig:long}.

\begin{figure}[b]
	\centering
	\subcaptionbox{short stick\label{fig:short}}{\includegraphics[width=0.18\textwidth,page=3]{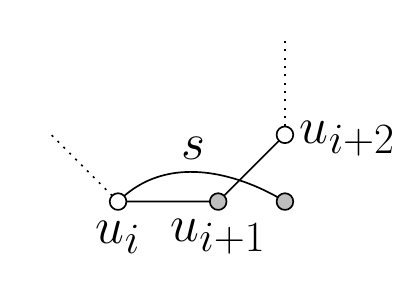}}
	\hfil
	\subcaptionbox{long stick\label{fig:long}}{\includegraphics[width=0.18\textwidth,page=4]{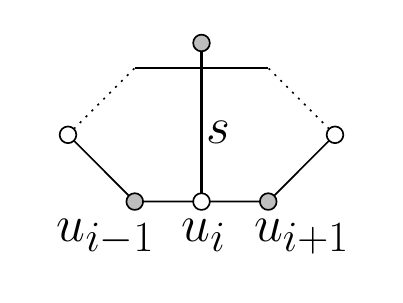}}
	\hfil
	\subcaptionbox{scissor + twin\label{fig:single-scissor}}{\includegraphics[scale=0.55,page=6]{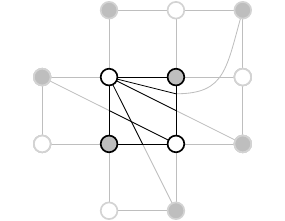}} 
	\hfil
	\subcaptionbox{pseudo-scissor\label{fig:pseudo-scissor}}{\includegraphics[scale=0.55,page=7]{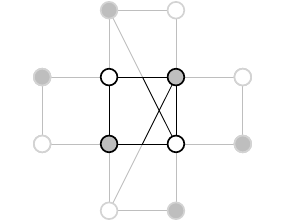}}
	\caption{%
		Illustration of sticks, scissors and twins.}
	\label{fig:q-sticks}
\end{figure}

In the following, we will assume that among all optimal bipartite $2$-planar graphs with~$n$ vertices, $G$ is chosen such that its planar structure $G_p$ is the densest among the planar structures of all other optimal bipartite $2$-planar graphs with $n$ vertices; we call $G_p$ \emph{maximally dense}. Under this assumption, we first prove that $G_p$ is a spanning quadrangulation (Lemma~\ref{lem:sticks} in Section~\ref{subsec:twoplanar-details}). For this, we first show that $G_p$ is connected (Lemma~\ref{lem:connected} in Section~\ref{subsec:twoplanar-details}), as otherwise it is always possible to augment it by adding an edge joining two connected components of it. Then, we show that all faces of $G_p$ are of length four. Our proof by contradiction is rather technical; assuming that there is a face $f$ with length greater than four in $G_p$, we consider two main cases: %
\begin{inparaenum}[(i)]
\item $f$ contains no sticks, but middle-parts, and 
\item $f$ contains at least one stick.  
\end{inparaenum}
With a careful case analysis, we lead to a contradiction either to the maximality of $G_p$ or to the fact that $G$ is optimal.

Since $G_p$ is a quadrangulation, it has exactly $2n-4$ edges and $n-2$ faces. Our goal is to prove that the average number of sticks for a face is at most $3$. Since the number of edges of $G \setminus G_p$ equals half the number of sticks over all faces of $G_p$, this implies that $G$ cannot have more than $2n-4+\frac{3}{2}(n-2)=3.5n-7$ edges, which gives the desired upper bound. 

Let $f$ be a face of $G_p$. Denote by $h(f)$ the number of sticks contained in~$f$. A \emph{scissor}~of~$f$ is a pair of crossing sticks starting from non-adjacent vertices of $f$, while a \emph{twin} of $f$ is a pair of sticks starting from the same vertex of $f$ crossing the same boundary edge of $f$; see Fig.~\ref{fig:single-scissor}. We refer to a pair of crossing sticks starting from adjacent vertices of $f$ as a \emph{pseudo-scissor}; see Fig.~\ref{fig:pseudo-scissor}. Next, we show that a face of $G_p$ contains a maximum~number of sticks (that is, $4$) only in the presence of scissors or twins, due~to~$2$-planarity (see Lemma~\ref{lem:quad-face} in Section~\ref{subsec:twoplanar-details}).    

An immediate consequence of the aforementioned property is that $h(f) \leq 3$, for every face $f$ containing a pseudo-scissor (Corollary~\ref{cor:pseudo-scissor} in Section~\ref{subsec:twoplanar-details}). We now consider specific ``neighboring'' faces of a face $f$ of $G_p$ with four sticks and prove that they cannot contain so many sticks. Observe that each edge corresponding to a stick of $f$ starts from a vertex of $f$ and ends at a vertex of another face of $G_p$. We call this other face, a \emph{neighbor} of this stick. The set of neighbors of the sticks forming a scissor (twin) of $f$ form the so-called \emph{neighbors} of this scissor (twin). 
Since $h(f)=4$, face $f$ contains two sticks $s_1$ and $s_2$ forming a twin or a scissor, with neighbors $f_1$ and $f_2$. By $2$-planarity and based on a technical case analysis, we show that $h(f_1) + h(f_2) \leq 7$ except for a single case, called \emph{$8$-sticks configuration} and illustrated in Fig.~\ref{fig:8stick-conf}, for which $h(f_1) + h(f_2) = 8$ (refer to Lemmas~\ref{lem:sticks-of-neighbours-1}--\ref{lem:sticks-of-neighbours-2} in Section~\ref{subsec:twoplanar-details}).

\begin{figure}
	\centering
	\subcaptionbox{\label{fig:8stick-conf}}{\includegraphics[height=2.1cm,page=8]{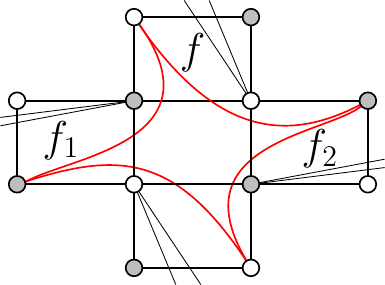}}
	\hfil
	\subcaptionbox{\label{fig:8stick-elim}}{\includegraphics[height=2.1cm,page=9]{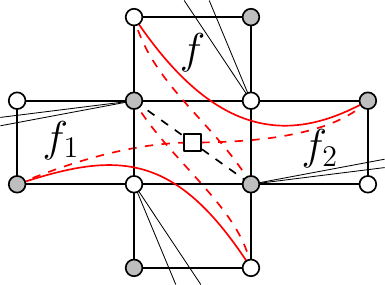}}
	\caption{%
	Illustration of 
	(a)~the $8$-sticks configuration, and
	(b)~its elimination.}
	\label{fig:q-half-edges}
\end{figure}

Assume first that $G$ does not contain any $8$-sticks configuration. We introduce an auxiliary graph $H$, which we call \emph{dependency graph}, having a vertex for each face of $G_p$. Then, for each face $f$ of $G_p$ containing a scissor or a twin with neighbors $f_1$ and $f_2$, such that $h(f_1) \leq h(f_2)$, graph $H$ has an edge from $f$ to $f_1$; note that $f_1=f_2$ is~possible. 
To prove that the average number of sticks for a face of $G_p$ is at most $3$ (which implies the desired upper bound), it suffices to prove that the number of faces of $G_p$ that contain two sticks is at least as large as the number of faces that contain four sticks. The latter is guaranteed by the following facts for every face $f$ of $G_p$: %
\begin{inparaenum}[(i)]
\item if $h(f)=4$, then $f$ has two outgoing edges and no incoming edge in $H$,
\item if $h(f)=3$, then the number of outgoing edges of $f$ in $H$ is at least as large as the number of its incoming edges, and finally
\item if $h(f)=2$, then $f$ has at most two incoming edges in $H$
\end{inparaenum}  
(see Properties~\ref{prp:4sticks}, \ref{prp:3sticks} and~\ref{prp:2sticks} in Section~\ref{subsec:twoplanar-details}). Hence, if $G$ does not contain any $8$-sticks configuration, then $G$ has at most $3.5n-7$ edges.

To complete the proof, assume now that $G$ contains $8$-sticks configurations. We eliminate each of them (without introducing new ones) by adding one vertex, and by replacing two edges of $G$ by six other edges violating neither its bipartiteness nor its $2$-planarity, as in Fig.~\ref{fig:8stick-elim}. Note that the derived graph $G'$ has a planar structure that is a spanning quadrangulation not containing any $8$-sticks configuration. Since $G'$ has one vertex and four edges more than $G$ for each $8$-sticks configuration and since the vertices of $G'$ have degree at most $3.5$ on average, by reversing the augmentation steps we can conclude that $G$ cannot have a larger edge density than $G'$. This implies the main results of this section, that is, a bipartite $n$-vertex $2$-planar multigraph has at most $3.5n - 7$ edges (see Theorem~\ref{thm:2-upper} in Section~\ref{subsec:twoplanar-details}).

\subsection{The details of our approach}
\label{subsec:twoplanar-details}

In this subsection, we give the details of our approach. We start by proving that a maximally dense planar structure of an optimal bipartite $2$-planar graph is connected.

\begin{lemma}\label{lem:connected}
Let $G$ be an optimal bipartite $2$-planar graph, such that its planar structure $G_p$ is maximally dense. Then, the planar structure $G_p$ of $G$ is connected.
\end{lemma}
\begin{proof}
Suppose, for a contradiction, that $G_p$ is not connected. Since $G$ is assumed to be connected, there exists an edge $(u,v)$ in $G \setminus G_p$ such that $u$ and $v$ belong to two different connected components $c_u$ and $c_v$ of $G_p$, respectively. Note that $(u,v)$ is crossed by at least an edge $(u_1,u_2)$ of $G_p$; we assume without loss of generality that the crossing with $(u_1,u_2)$ is the first one that is encountered when walking along $(u,v)$ from $u$ to $v$. Then, $(u,v)$ may be crossed by another edge $(v_1,v_2)$, which may belong to $G_p$ or to $G \setminus G'$. We assume that $u_1$ ($v_1$) does not belong to the same partition as $u$ (as $v$), while $u_2$ ($v_2$) does. This implies that $(u,u_1)$, $(v,v_1)$, $(u_1,v_1)$, and $(u_2,v_2)$ may be added to $G$ without violating~bipartiteness.

\begin{figure}[h]
	\centering
	\subcaptionbox{\label{fig:2-planar-connectivity-2}}{\includegraphics[width=0.16\textwidth,page=2]{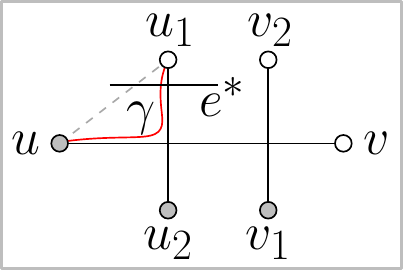}}
	\hfil
	\subcaptionbox{\label{fig:2-planar-connectivity-3}}{\includegraphics[width=0.16\textwidth,page=3]{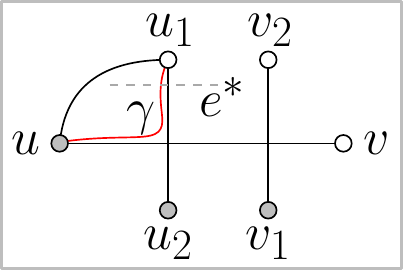}}
	\hfil
	\subcaptionbox{\label{fig:2-planar-connectivity-4}}{\includegraphics[width=0.16\textwidth,page=4]{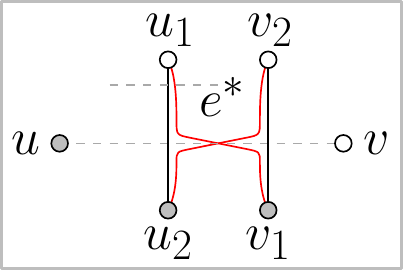}}
	\hfil
	\subcaptionbox{\label{fig:2-planar-connectivity-5}}{\includegraphics[width=0.16\textwidth,page=5]{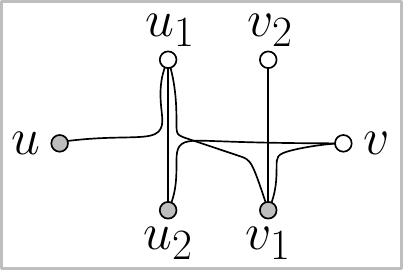}}
	\caption{Augmentation of $G$ into a connected $G_p$.}
	\label{fig:connectivity}
\end{figure}

Suppose first that $u_1$ does not belong to $c_u$, which implies that edge $(u,u_1)$ does not belong to $G_p$.
Consider a curve $\gamma$ from $u$ to $u_1$ that first follows $(u,v)$ till its intersection point with $(u_1,u_2)$, and then follows this edge till $u_1$. Note that the first part of $\gamma$ does not cross any edge, while the second one crosses at most one edge, call it $e^*$; see Fig.~\ref{fig:2-planar-connectivity-2}.
If $\gamma$ does not cross any edge, then we can add edge $(u,u_1)$ to $G_p$, contradicting either the optimality of $G$ or the fact that $G$ has been chosen so that $G_p$ is the densest possible.
If $\gamma$ crosses $e^*$, then observe that $e^*$ belongs to $G \setminus G_p$, since $e^*$ crosses $(u_1,u_2)$ which belongs to $G_p$. If $(u,u_1)$ does not belong to $G$, then we draw it as $\gamma$, we add it to $G_p$, and we remove $e^*$ from $G$. If $(u,u_1)$ belongs to $G \setminus G_p$ and crosses $e^*$, then we redraw $(u,u_1)$ as $\gamma$ and add it to $G_p$, which leads to a contradiction the fact that $G$ has been chosen so that $G_p$ is the densest possible. Finally, if $(u,u_1)$ belongs to $G \setminus G_p$ and does not cross $e^*$, then we draw a copy of $(u,u_1)$ as $\gamma$ and add it to $G_p$; since the two endvertices of $e^*$ lie in different regions delimited by the two copies of $(u,u_1)$, these two copies are non-homotopic; see Fig.~\ref{fig:2-planar-connectivity-3}. The contradiction is again due to the fact that $G$ has been chosen so that $G_p$ is densest.

Since in all the cases we have a contradiction, this completes the analysis of the case in which $u_1$ does not belong to $c_u$. So, in the following we will assume that $u_1$ belongs to $c_u$.

Note that, if $(u,v)$ crosses only one edge, i.e., $(u_1,u_2)$, we can use the same argument to prove that $u_2$ belongs to $c_v$. However, since $(u_1,u_2)$ belongs to $G_p$, it follows that $c_u=c_v$; a contradiction. Thus, it only remains to consider the case in which $(u,v)$ also crosses~$(v_1,v_2)$.

Suppose first that $(v_1,v_2)$ belongs to $G_p$. As before, we can assume that $v_1$ belongs to $c_v$. Note that neither $(u_1,v_1)$ nor $(u_2,v_2)$ belong to $G_p$, as otherwise $c_u=c_v$ would hold.
Consider a curve $\gamma_1$ from $u_1$ to $v_1$ that follows edges $(u_1,u_2)$, then $(u,v)$, and finally $(v_1,v_2)$. Also, consider a curve $\gamma_2$ from $u_2$ to $v_2$ that follows edges $(u_1,u_2)$, then $(u,v)$, and finally $(v_1,v_2)$ (in Fig.~\ref{fig:2-planar-connectivity-4} curves $\gamma_1$ and $\gamma_2$ are colored red). Note that the parts of these curves following $(u,v)$ do not cross any edge, while the other parts cross at most one edge each. However, if the part of $\gamma_1$ following edge $(u_1,u_2)$ crosses an edge, then the part of $\gamma_2$ following edge $(u_1,u_2)$ does not cross any edge, and the same holds for the parts following $(v_1,v_2)$. This implies that at least one of $\gamma_1$ and $\gamma_2$, say $\gamma_1$, crosses at most two edges, namely $(u,v)$ and an edge $e^*$ that is also crossed by either $(u_1,u_2)$ or $(v_1,v_2)$, say $(u_1,u_2)$; see Fig.~\ref{fig:2-planar-connectivity-4}. Note that both $(u,v)$ and $e^*$ belong to $G \setminus G_p$. We remove both these edges from $G$, and we add (a non-homotopic copy of) edge $(u_1,v_1)$ to $G_p$, drawing it as $\gamma_1$, and (a non-homotopic copy of) edge $(u_2,v)$ to $G \setminus G_p$, drawing it by following $(u,v)$ and $\gamma_1$; see Fig.~\ref{fig:2-planar-connectivity-5}. Note that $(u_1,v_1)$ only crosses $(u_2,v)$, while $(u_2,v)$ crosses $(u_1,v_1)$ and $(v_1,v_2)$; since this latter edge was crossing $(u,v)$ before it was removed, it still has at most two crossings.
Since we replaced two edges of $G \setminus G_p$ with one of $G_p$ and one of $G \setminus G_p$, we have again a contradiction. This concludes the case in which $(v_1,v_2)$ belongs to $G_p$.

Suppose finally that $(v_1,v_2)$ belongs to $G \setminus G_p$. We remove $(u,v)$ from $G \setminus G_p$ and add edge $(v,u_2)$ to $G_p$, drawing it as a curve following $(u,v)$ and $(u_1,u_2)$. If this curve crosses an edge $e^* \in G \setminus G_p$ that is also crossed by $(u_1,u_2)$, then we remove $e^*$ from $G$ and add to $G \setminus G_p$ the edge out of $(u_1,v_1)$ and $(u_2,v_2)$ that can be drawn without crossing any edge other than (possibly) $(v,u_2)$. This completes our case analysis and thus concludes the proof of this lemma.
\end{proof}

In the following lemma, we are proving that a maximally dense planar structure of an optimal bipartite $2$-planar graph is a quadrangulation, i.e., a planar graph whose faces are of length four. Since the proof of this lemma is rather technical and it requires a careful case analysis, we have decided to postpone its proof for Section~\ref{subsec:sticksproof} in order to keep the flow of the proof of our main theorem clear.

\newcommand{\quadrangulation}{Let $G$ be an optimal bipartite $2$-planar graph, such that its planar structure $G_p$ is maximally dense. Then, the planar structure $G_p$ of $G$ is a quadrangulation.}
\begin{lemma}\label{lem:sticks}
\quadrangulation
\end{lemma}
\begin{proof}
Refer to Section~\ref{subsec:sticksproof}
\end{proof}

Next, we show that a face of $G_p$ contains a maximum~number of sticks (that is, $4$) only in the presence of scissors or twins.

\begin{lemma}\label{lem:quad-face}
Let $G$ be an optimal bipartite $2$-planar graph, such that its planar structure $G_p$ is maximally dense. Then, for each face $f$ of $G_p$, it holds $h(f) \leq 4$. Further, if $h(f) = 4$, then $f$ contains one of the following: two scissors, or two twins, or a scissor and a twin.
\end{lemma}
\begin{proof}
Let $f = (u_1,u_2,u_3,u_4)$. We first prove the statement under the assumption that~there exists a vertex of $f$, say $u_1$, that has at least two sticks of $f$. Since $f$ has four vertices, by the pigeonhole principle this assumption is without loss of generality when $h(f) > 4$.
Since the sticks of $u_1$ cannot cross edges incident to $u_1$, they either both cross the same edge of $f$, say w.l.o.g.~$(u_2,u_3)$, or one of them crosses $(u_2,u_3)$ and the other one crosses $(u_3,u_4)$ of $f$.

\begin{figure}
	\centering
	\subcaptionbox{\label{fig:quad-face-1}}{\includegraphics[scale=0.55,page=1]{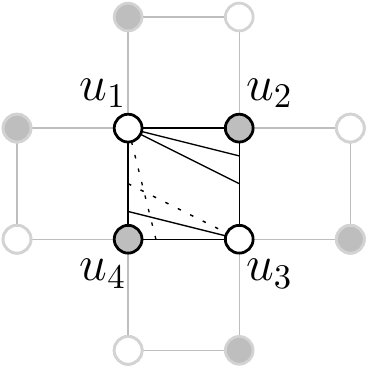}}
	\hfil
	\subcaptionbox{\label{fig:quad-face-2}}{\includegraphics[scale=0.55,page=2]{atmostfoursticks}}
	\hfil
	\subcaptionbox{\label{fig:quad-face-3}}{\includegraphics[scale=0.55,page=3]{atmostfoursticks}}
	\caption{Cases for the proof of Lemma~\ref{lem:quad-face}.}
	\label{fig:q-sticks}
\end{figure}

We first consider the former case. Since $(u_2,u_3)$ has already two crossings, it cannot have any other crossing, by $2$-planarity. Further, $u_2$ cannot have any stick in $f$, as otherwise such a stick would cross both sticks of $u_1$, plus a boundary edge of $f$, contradicting $2$-planarity. For the same reason, neither $u_3$ nor $u_4$ can have sticks crossing $(u_1,u_2)$. Hence, $u_4$ has no sticks. Therefore, the remaining sticks of $f$ either start at $u_3$ and cross $(u_1,u_4)$, or start at $u_1$ and cross $(u_3,u_4)$. If there are at least three additional sticks (and thus $h(f)>4$), either $(u_1,u_4)$ or $(u_3,u_4)$ have at least three crossings, contradicting $2$-planarity. Further, if they are exactly two (and thus $h(f)=4$), they form either a scissor or a twin; see Fig.~\ref{fig:quad-face-1}. Since the other two sticks of $u_1$ form a twin, the statement of the lemma holds in this case.

We now consider the later case, in which one stick of $u_1$ crosses $(u_2,u_3)$ and the other one crosses $(u_3,u_4)$. Note that, if $u_1$ has a third stick, then the previous case applies. So, we may assume w.l.o.g.~that $u_1$ has exactly two sticks, i.e., the one crossing $(u_2,u_3)$ and the other one crossing $(u_3,u_4)$. Also, note that there is no stick of $u_2$ crossing $(u_1,u_4)$, and no stick of $u_4$ crossing $(u_1,u_2)$, as otherwise these sticks would cross both sticks of $u_1$, plus a boundary edge of $f$, contradicting $2$-planarity.

Suppose now that there exists a stick of $u_2$ crossing $(u_3,u_4)$. Then, there is no other stick of $u_2$, since $(u_3,u_4)$ is already crossed twice. Also, there is no stick of $u_3$ and no stick of $u_4$, since any of these sticks would cross a stick of $u_1$, a stick of $u_2$, and a boundary edge of $f$. This implies that $h(f) \leq 3$ in this case.
Analogously, we can prove that if there exists a stick of $u_4$ crossing $(u_2,u_3)$, then $h(f) \leq 3$.

Since additional sticks starting from $u_1$ are ruled out by the previous case, all remaining sticks contained in $f$ have to be incident to $u_3$. If we assume that $h(f) > 4$, then at least two sticks of $u_3$ would have to cross the same boundary edge of $f$ and the same stick of $u_1$ (that is, either $s_1$ or $s_2$), which is a contradiction to $2$-planarity. Hence, $h(f) \leq 4$ holds, as desired. Consider now the case where $h(f) = 4$ holds. In this case, the two sticks of $u_3$ have to cross different boundary edges of $f$ forming two scissors with $s_1$ and $s_2$; see Fig.~\ref{fig:quad-face-2}. Thus, the statement holds also in this case.

We now remove the assumption that there exists a vertex with two sticks. This directly implies $h(f) \leq 4$; also, if $h(f) = 4$, then each vertex of $f$ has exactly one stick. To conclude the proof of the statement, it suffices to show that this case is not possible. Let $s_i$ be the stick of $u_i$, for $i=1,2,3,4$. We first observe that stick $s_i$, must cross either $s_{i-1}$ or $s_{i+1}$. To see this consider, e.g., stick $s_1$. If $s_1$ crosses edge $(u_2,u_3)$ of $f$, then it also crosses $s_2$; if it crosses edge $(u_3,u_4)$ of $f$, then it also crosses $s_4$. Since by $2$-planarity a stick cannot be crossed by two other sticks, as it also crosses a boundary edge of $f$, the only configuration we have to consider is the one in which there is a crossing between $s_1$ and $s_2$, and one between $s_3$ and $s_4$ (or any other symmetric configuration). Note that, in this case, both $s_1$ and $s_4$ cross $(u_2,u_3)$; since they also have a crossing with a stick inside $f$, the edges corresponding to them must end at vertices of the face $f'$ of $G_p$ sharing edge $(u_2,u_3)$ with $f$. However, due to bipartiteness, these two edges must cross with each other inside $f'$ in order to reach their end-vertices, contradicting $2$-planarity; see Fig.~\ref{fig:quad-face-3}. This concludes the proof of this lemma.
\end{proof}

\begin{corollary}\label{cor:pseudo-scissor}
Let $G$ be an optimal bipartite $2$-planar graph, such that its planar structure $G_p$ is maximally dense. If a face $f$ of $G_p$ contains a pseudo-scissor, then $h(f) \leq 3$.
\end{corollary}

\noindent The following three lemmas are also consequence of Lemma~\ref{lem:quad-face}.

\begin{lemma}\label{lem:forced-twin}
Let $G$ be an optimal bipartite $2$-planar graph, such that its planar structure $G_p$ is maximally dense and let $f$ be a face of $G_p$ with $h(f) = 4$. If the edge corresponding to a stick $s$ contained in $f$ is crossed outside of $f$, then $s$ is part of a twin in $f$.
\end{lemma}
\begin{proof}
By Lemma~\ref{lem:quad-face} and since $h(f)=4$, stick $s$ is part of either a twin or a scissor in $f$. If $s$ is part of a scissor, then $e$ has a crossing inside $f$, a crossing with the boundary edge of $f$ and a crossing outside $f$, which contradicts $2$-planarity. Hence, $s$ is part of a twin in $f$.
\end{proof}

\begin{lemma}\label{lem:2-face}
Let $G$ be an optimal bipartite $2$-planar graph, such that its planar structure $G_p$ is maximally dense and let $f$ be a face of $G_p$ that contains two sticks $s_1$ and $s_2$ of adjacent vertices of $f$ such that $s_1$ and $s_2$ cross the same boundary edge of $f$. Let $e_1$ and $e_2$ be the edges corresponding to $s_1$ and $s_2$. If either $e_1$ and $e_2$ cross each other inside $f$ or each of $e_1$ and $e_2$ has an additional crossing outside $f$, then $h(f) = 2$.
\end{lemma}
\begin{proof}
Let $f = (u_1,u_2,u_3,u_4)$. W.l.o.g.~assume that $s_1$ is a stick of $u_1$ and $s_2$ is a stick of $u_2$. Then, $e_1$ and $e_2$ cross edge $(u_3,u_4)$ of $f$. Note that in both cases of the lemma, sticks $s_1$ and $s_2$ cannot be crossed by any other stick inside $f$ due to $2$-planarity. This implies that there are no sticks incident to $u_3$ and $u_4$. Furthermore, any other stick of $u_1$ would cross either stick $s_2$ or edge $(u_3,u_4)$ of $f$. Symmetrically, any other stick of $u_2$ would cross either stick $s_1$ or edge $(u_3,u_4)$ of $f$. In both cases, $2$-planarity is violated. Therefore, $h(f) = 2$.
\end{proof}

\begin{lemma}\label{lem:2-middle}
Let $G$ be an optimal bipartite $2$-planar graph, such that its planar structure $G_p$ is maximally dense and let $f$ be a face of $G_p$ that contains two middle-parts $m_1$ and $m_2$ crossing the same boundary edge of $f$. Then $h(f) \leq 3$.
\end{lemma}
\begin{proof}
Let $f = (u_1,u_2,u_3,u_4)$ and assume w.l.o.g.~that $(u_1,u_2)$ is the edge of $f$ crossed by $m_1$ and $m_2$. For a proof by contradiction, assume that $h(f) > 3$. Then, by Lemma~\ref{lem:quad-face} it follows that $h(f)=4$ and therefore $f$ contains two scissors or two twins or a scissor and a twin. Since $(u_1,u_2)$ is already involved in two crossings with $m_1$ and $m_2$, $f$ cannot contain two scissors. Hence, it contains at least one twin; call it~$\tau$. Suppose first that $\tau$ is incident to vertex $u_1$ of $f$. Then, the two edges of $\tau$ cannot cross $(u_2,u_3)$, as otherwise they would either cross $m_1$ and $m_2$ or they would introduce more than two crossings along $(u_2,u_3)$. Thus, the two edges of $\tau$ cross $(u_3,u_4)$ while $m_1$ and $m_2$ cross $(u_2,u_3)$. In this configuration no other edge can be added to $f$ without violating $2$-planarity, which implies that $h(f)=2$; a contradiction to our initial assumption that $h(f)=4$. Symmetrically, $\tau$ is not incident to $u_2$. To complete the proof of this lemma, consider the case where $\tau$ is incident to $u_3$; the case where $\tau$ is incident to $u_4$ is symmetric. In this case, the edges of $\tau$ cross edge $(u_1,u_4)$ of $f$ and thus $m_1$ and $m_2$ have to cross edge $(u_2,u_3)$ of $f$. Hence, we reached again a configuration in which no other edge can be added to $f$ without violating $2$-planarity, which implies that $h(f)=2$; a contradiction to our initial assumption that $h(f)=4$.
\end{proof}

In Lemmas~\ref{lem:sticks-of-neighbours-1} and~\ref{lem:sticks-of-neighbours-2}, we show that for a face containing a scissor $\sigma$ (a twin $\tau$), both neighbors of $\sigma$ (of $\tau$) cannot contain four sticks except for the special case of $8$-stick~configurations.

\begin{lemma}\label{lem:sticks-of-neighbours-1}
Let $G$ be an optimal bipartite $2$-planar graph, such that its planar structure $G_p$ is maximally dense and let $f$ be a face of $G_p$ that contains two sticks $s_1$ and $s_2$ forming a scissor $\sigma$ contained in $f$. Let $f_1$ and $f_2$ be the neighbors of scissor $\sigma$. Then, $h(f_1) + h(f_2) \leq 7$.
\end{lemma}
\begin{proof}
Let $e_1$ and $e_2$ be the edges corresponding to sticks $s_1$ and $s_2$ of $f$. Since $s_1$ and $s_2$ form scissor $\sigma$ in~$f$, both the stick of $e_1$ in $f_1$ and the stick of $e_2$ in $f_2$ are crossing-free. If either $h(f_1) \leq 3$ or $h(f_2) \leq 3$, then the statement follows by Lemma~\ref{lem:quad-face}. So, it remains to rule out the case where $h(f_1)>3$ and $h(f_2)>3$. Assume w.l.o.g.~that $h(f_1)>3$. By Lemma~\ref{lem:forced-twin}, the stick corresponding to $e_1$ in $f_1$ is part of a twin (recall that the stick of $e_1$ in $f_1$ is crossing-free). Thus, the edge $e_1'$ corresponding to the other stick of the twin crosses the edge shared by $f$ and $f_1$; see Fig.~\ref{fig:scissor-variations-2}. Note that $e_1'$ cannot end in $f$ and, in particular, it must have a middle-part in $f$ and end in $f_2$, as otherwise it would cross~$e_2$. Due to bipartiteness, it is incident to a different vertex than the endpoint of $e_2$ in $f_2$. By Lemma~\ref{lem:2-face}, it follows that $h(f_2) = 2$ and thus $h(f_1)+h(f_2) \leq 7$. This completes~the~proof.
\end{proof}

\begin{figure}
	\centering
	\subcaptionbox{\label{fig:scissor-variations-2}}{\includegraphics[width=0.24\textwidth,page=2]{scissor-variations}}
	\hfill
	\subcaptionbox{\label{fig:twin-variations-1}}{\includegraphics[width=0.24\textwidth,page=1]{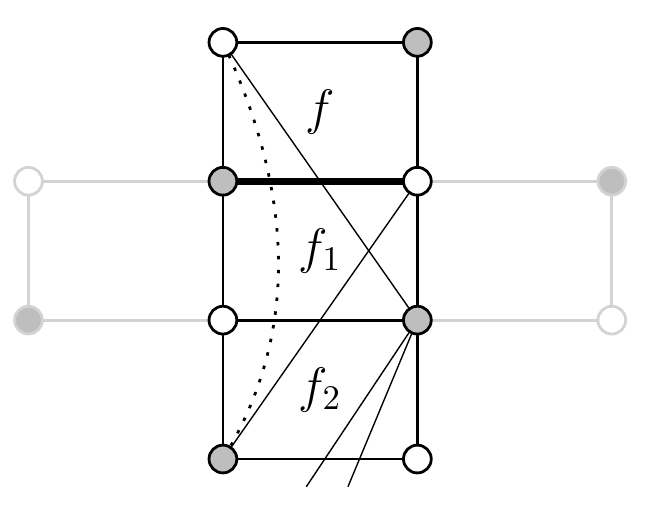}}
	\hfill
	\subcaptionbox{\label{fig:twin-variations-2}}{\includegraphics[width=0.24\textwidth,page=2]{twin-variations}}
	\hfill
	\subcaptionbox{\label{fig:twin-variations-3}}{\includegraphics[width=0.24\textwidth,page=3]{twin-variations}}
	\hfill	
	\subcaptionbox{\label{fig:twin-variations-4}}{\includegraphics[width=0.24\textwidth,page=4]{twin-variations}}
	\hfill
	\subcaptionbox{\label{fig:twin-variations-5}}{\includegraphics[width=0.24\textwidth,page=5]{twin-variations}}
	\hfill
	\subcaptionbox{\label{fig:special-case-1}}{\includegraphics[width=0.24\textwidth,page=6]{twin-variations}}
	\hfill
	\subcaptionbox{\label{fig:special-case-2}}{\includegraphics[width=0.24\textwidth,page=7]{twin-variations}}
	\caption{%
		Different configurations used in the proof of (a)~Lemma~\ref{lem:sticks-of-neighbours-1} and (b)-(g)~Lemma~\ref{lem:sticks-of-neighbours-2}.}
	\label{fig:q-half-edges}
\end{figure}

\begin{lemma}\label{lem:sticks-of-neighbours-2}
Let $G$ be an optimal bipartite $2$-planar graph, such that its planar structure $G_p$ is maximally dense and let $f$ be a face of $G_p$ that contains two sticks $s_1$ and $s_2$ forming a twin $\tau$ contained in $f$. Let $f_1$ and $f_2$ be the neighbors of twin $\tau$. Then, $h(f_1) + h(f_2) \leq 7$ except for a single case illustrated in Fig.~\ref{fig:special-case-1} for which $h(f_1) + h(f_2) = 8$. We refer to this exceptional case as \emph{$8$-sticks configuration}.
\end{lemma}
\begin{proof}
Let $e_1$ and $e_2$ be the edges corresponding to sticks $s_1$ and $s_2$ of $f$. Let also $e$ be the boundary edge of $f$ that is crossed by both $s_1$ and $s_2$ (bold in Figs.~\ref{fig:twin-variations-1}-\ref{fig:special-case-2}). As in the previous lemma, if either $h(f_1) \leq 3$ or $h(f_2) \leq 3$, then by Lemma~\ref{lem:quad-face} the statement trivially follows. So, it remains to rule out the case where $h(f_1)>3$ and $h(f_2)>3$. We distinguish two cases based on whether one of $e_1$ and $e_2$ ends in the face sharing edge $e$ with $f$, or not.

Consider first the case where one of $e_1$ and $e_2$, say $e_1$, ends in this face (which thus coincides with $f_1$; see Figs~\ref{fig:twin-variations-1}-\ref{fig:twin-variations-3}). Let $s_1'$ be the stick of $e_1$ in $f_1$. Note that $e_2$ has a middle-part in $f_1$, as otherwise $e_1$ and $e_2$ would be two homotopic copies of the same edge (by bipartiteness). Then, $e_2$ can end in one of the three faces adjacent to $f_1$ that are not identified with $f$ (see the dotted edge in Figs~\ref{fig:twin-variations-1}-\ref{fig:twin-variations-3}), that is, $f_2$ is one of these three faces. Since $e_2$ is crossed twice while entering $f_2$ and since $h(f_2) = 4 $, by Lemma~\ref{lem:forced-twin} the stick $s_2'$ in $f_2$ corresponding to $e_2$ belongs to a twin. The edge $e_2'$ corresponding to the stick forming this twin different from $s_2'$ crosses the boundary edge shared by $f_2$ and $f_1$. First note that $e_2'$ cannot end in $f$ as otherwise it would create a homotopic multiedge with $e_2$. Suppose now that $e_2'$ ends in $f_1$. Due to bipartiteness, its endpoint in $f_1$ is adjacent to the endpoint of $e_1$ in $f_1$, and $e_1$ and $e_2$ cross inside $f_1$. Thus, $f_1$ contains a pseudo-scissor, which by Corollary~\ref{cor:pseudo-scissor} implies that $h(f_1) \leq 3$ and the statement follows.
Finally, assume that $e_2'$ ends in a face different from $f$ and $f_1$, which implies that $f_1$ contains two middle-parts (of $e_2$ and $e_2'$) crossing exactly the same edge of $f_1$; by Lemma~\ref{lem:2-middle} $h(f_1) \leq 3$ and the statement follows.

We now consider the case in which $e_1$ and $e_2$ end in faces not adjacent to $f$. Let $f'$ be the face that shares edge $e$ with $f$; see Figs.~\ref{fig:twin-variations-4}-\ref{fig:special-case-1}. We distinguish two subcases based on whether edge $e_1$ crosses $e$ and the edge of $f'$ that is independent to $e$ or not. We first consider the former case; see Figs.~\ref{fig:twin-variations-4}-\ref{fig:twin-variations-5}. Then, edge $e_2$ must end at a face $f_2$ that is adjacent to $f'$ and is identified neither with $f$ nor with $f_1$. Since $h(f_2)=4$, by Lemma~\ref{lem:forced-twin} the stick $s_2'$ in $f_2$ corresponding to $e_2$ belongs to a twin. Let $e_2'$ be the edge corresponding to the stick forming this twin different from $s_2'$. Edge $e_2'$ crosses the boundary edge shared by $f'$ and $f_2$ (by definition of twin). First note that $e_2'$ cannot end in $f$ as otherwise it would create a homotopic multiedge with $e_2$. We now explicitly consider the two cases of Figs.~\ref{fig:twin-variations-4}-\ref{fig:twin-variations-5}. In the case illustrated in Fig.~\ref{fig:twin-variations-4}, edge $e_2'$ clearly violates $2$-planarity (refer to the red-colored edge of Fig.~\ref{fig:twin-variations-4}). In the case illustrated in Fig.~\ref{fig:twin-variations-5}, edge $e_2$ has to end in $f_1$, and by Lemma~\ref{lem:2-face} it follows that $h(f_1)\leq 2$.

By symmetry, to complete the proof of our lemma, it remains to consider the case in which neither $e_1$ nor $e_2$ crosses $e$ and the edge of $f'$ that is independent to $e$; see Fig.~\ref{fig:special-case-1}. Since $h(f_1) = 4$, by Lemma~\ref{lem:forced-twin} stick $s_1'$ in $f_1$ corresponding to $e_1$ belongs to a twin. Symmetrically, stick $s_2'$ in $f_2$ corresponding to $e_2$ also belongs to a twin. As in the previous case, let $e_1'$ be the edge corresponding to the stick forming the twin in $f_1$ that is different from $s_1'$ and let $e_2'$ be the edge corresponding to the stick forming the twin in $f_2$ that is different from $s_2'$. By definition of twin, edge $e_1'$ crosses the edge shared by $f'$ and $f_1$, while edge $e_2'$ crosses the edge shared by $f'$ and $f_2$. Due to bipartiteness and $2$-planarity, $e_1'$ cannot end in $f'$. If $e_1'$ ends in $f_2$, then $e_2'$ cannot exist and $h(f_2)\leq 3$ and the statement follows. To reach $h(f_1) + h(f_2) = 8$, $e_1'$ has to end in face $f^*$, which is the face adjacent to $f'$ that is not identified with one of $f$, $f_1$ and $f_2$. Assume first that $e_2'$ does not end in $f^*$. Since $e_1'$ exists, $e_2'$ can end neither in $f'$ nor in $f_1$ due to $2$-planarity. Therefore, $h(f_2) \leq 3$ and the statement holds. For the final case, assume that $e_2'$ ends in $f^*$. By bipartiteness, $e_1'$ and $e_2'$ have to end in the same vertex and we get exactly the single special configuration of Fig.~\ref{fig:special-case-1}.

We conclude that for any case, except for the special case, $h(f_1) + h(f_2) \leq 7$ holds. Note that by Lemma~\ref{lem:quad-face}, $h(f_1) + h(f_2) \leq 8$ holds in general. In particular, for the special case it is not difficult to see that $h(f_1) + h(f_2) = 8$ holds, which is the reason why we refer to it as $8$-sticks configuration. This complete the proof of this lemma.
\end{proof}

We first prove our bound when $G$ does not contain $8$-sticks configurations. We will~extend the proof to the general case later. We introduce an auxiliary graph $H$, which we call \emph{dependency graph}, having a vertex for each face of $G_p$. Then, for each face~$f$ of $G_p$ containing a scissor or a twin with neighbors $f_1$ and $f_2$, such that $h(f_1) \leq h(f_2)$, graph $H$ has an edge from $f$ to $f_1$; note that $f_1=f_2$ is~possible. In the following, we study properties of~$H$.

\begin{prp}\label{prp:4sticks}
Every face $f$ of $G_p$ with $h(f)=4$ has two outgoing edges and no incoming edge in~$H$.
\end{prp}
\begin{proof}
The fact that $f$ has two outgoing edges follows from Lemma~\ref{lem:quad-face}. Suppose now that~$f$ has an incoming edge. Then, $f$ and one additional face $f_1$ are neighbors to a face~$f_2$. By Lemma~\ref{lem:sticks-of-neighbours-1} and \ref{lem:sticks-of-neighbours-2}, the neighbors of $f_2$ have at most $7$ sticks, thus $h(f_1) \leq 3$. But then $h(f) > h(f_1)$ and therefore $f$ has no incoming edge.
\end{proof}

\begin{prp}\label{prp:3sticks}
For every face $f$ of $G_p$ with $h(f)=3$, the number of outgoing edges is at least as large as the number of incoming edges.
\end{prp}
\begin{proof}

Since $h(f)=3$, $f$ cannot contain two scissors or two twins or a scissor and a twin. Based on this observation, we distinguish in our proof three main cases:

\begin{enumerate}[C.1]
\item \label{mc:scissor} $f$ contains a scissor,
\item \label{mc:twin}    $f$ contains a twin, and
\item \label{mc:neither} $f$ contains neither a scissor nor a twin.
\end{enumerate}

Observe that in Cases C.\ref{mc:scissor} and C.\ref{mc:twin}, $f$ has an outgoing edge in the dependency graph $H$ (by definition of~$H$). We will prove by contradiction that $f$ has at most one incoming edge in both cases. To this end, assume to the contrary that $f$ has more than one incoming edges. Under this assumption, at least one of the incoming edges is due to a stick belonging to the scissor or to the twin of $f$. Before we continue with the description of our approach, we first introduce some necessary notation. Let $s_1$ and $s_2$ be the sticks forming the scissor or the twin of $f$. Assume that $e_1$ and $e_2$ are the edges corresponding to $s_1$ and $s_2$. Let $f_1$ and $f_2$ be the neighbors of the scissor or the twin formed by $s_1$ and $s_2$. Let $s_1'$ and $s_2'$ be the other sticks of $e_1$ and $e_2$ contained in $f_1$ and $f_2$, respectively. In the following, we consider each of Cases C.\ref{mc:scissor} and C.\ref{mc:twin} separately.

\begin{figure}[p]
	\flushleft
    \subcaptionbox{\label{fig:3-face-incoming-1}}{\includegraphics[width=0.24\textwidth,page=1]{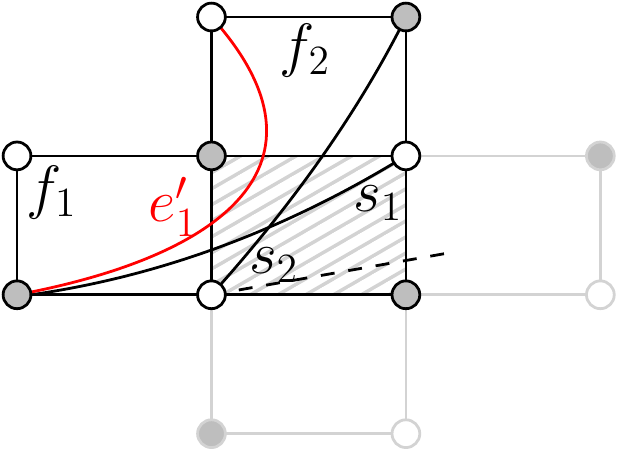}}
    \hfil
    \subcaptionbox{\label{fig:3-face-incoming-2}}{\includegraphics[width=0.24\textwidth,page=2]{3-face-incoming}}
    \hfil
    \subcaptionbox{\label{fig:3-face-incoming-3}}{\includegraphics[width=0.24\textwidth,page=3]{3-face-incoming}}
    \hfil
    \subcaptionbox{\label{fig:3-face-incoming-4}}{\includegraphics[width=0.24\textwidth,page=4]{3-face-incoming}}
    \hfil
    \subcaptionbox{\label{fig:3-face-incoming-5}}{\includegraphics[width=0.24\textwidth,page=5]{3-face-incoming}}
    \hfil
    \subcaptionbox{\label{fig:3-face-incoming-6}}{\includegraphics[width=0.24\textwidth,page=6]{3-face-incoming}}
    \hfil
    \subcaptionbox{\label{fig:3-face-incoming-7-1}}{\includegraphics[width=0.24\textwidth,page=7]{3-face-incoming}}
    \hfil
    \subcaptionbox{\label{fig:3-face-incoming-7-2}}{\includegraphics[width=0.24\textwidth,page=15]{3-face-incoming}}
    \hfil
    \subcaptionbox{\label{fig:3-face-incoming-8}}{\includegraphics[width=0.24\textwidth,page=8]{3-face-incoming}}
    \hfil
    \subcaptionbox{\label{fig:3-face-incoming-9}}{\includegraphics[width=0.24\textwidth,page=9]{3-face-incoming}}
    \hfil
    \subcaptionbox{\label{fig:3-face-incoming-10}}{\includegraphics[width=0.24\textwidth,page=10]{3-face-incoming}}
    \hfil
    \subcaptionbox{\label{fig:3-face-incoming-11}}{\includegraphics[width=0.24\textwidth,page=11]{3-face-incoming}}
    \hfil
    \subcaptionbox{\label{fig:3-face-incoming-12}}{\includegraphics[width=0.24\textwidth,page=12]{3-face-incoming}}
    \hfil
    \subcaptionbox{\label{fig:3-face-incoming-13}}{\includegraphics[width=0.24\textwidth,page=13]{3-face-incoming}}
    \hfil
    \subcaptionbox{\label{fig:3-face-incoming-14}}{\includegraphics[width=0.24\textwidth,page=14]{3-face-incoming}}
    \hfil\hfil\hfil\hfil\hfil\hfil\hfil\hfil\hfil\hfil\hfil\hfil\hfil
    \hfil\hfil\hfil\hfil\hfil\hfil\hfil\hfil\hfil\hfil\hfil\hfil\hfil
    \caption{
    Different configurations used in the proof of Cases~C.\ref{mc:scissor} and~C.\ref{mc:twin} of Property~\ref{prp:3sticks}; the gray-shaded region corresponds to face $f$ in our proof.}
    \label{fig:3-face-incoming}
 \end{figure}

We fist focus on Case~\ref{mc:scissor}. Here, $s_1$ and $s_2$ form a scissor contained in $f$; see Fig.~\ref{fig:3-face-incoming-1}. Since $e_1$ and $e_2$ cross in $f$, sticks $s_1'$ and $s_2'$ have to be crossing-free in $f_1$ and $f_2$, respectively. Therefore, each of $s_1'$ and $s_2'$ cannot form a scissor in $f_1$ and $f_2$, respectively. It follows that at least one of $s_1'$ and $s_2'$ belongs to a twin of $f_1$ or $f_2$, respectively (recall our initial assumption; $f$ has more than one incoming edges). Assume w.l.o.g.~that $s_1'$ is this stick (note that in this case $(f_1,f) \in E[H]$). Let $e_1'$ be the edge containing the stick of this twin that is different from $s_1'$ (red-colored in Fig.~\ref{fig:3-face-incoming-1}). By $2$-planarity, edge $e_1'$ has a middle part in $f$ and its other stick in $f_2$. In addition, due to bipartiteness the stick of $e_1'$ in $f_2$ ends at a different vertex than stick $s_2'$ ends. By Lemma~\ref{lem:2-face}, $h(f_2) = 2$ ; see Fig.~\ref{fig:3-face-incoming-1}, thus $h(f_2) < h(f)$ holds and by the definition of $H$ we have that $(f_1,f_2)\in E[H]$, which implies that $(f_1,f)\notin E[H]$. This is a contradiction to our initial assumption.

We next focus on Case~\ref{mc:twin}, which is a bit more involved. Here, $s_1$ and $s_2$ form a twin contained in $f$; see, e.g., Fig.\ref{fig:3-face-incoming-2}. As in the previous case, assume that $e_1$ and $e_2$ are the edges corresponding to $s_1$ and $s_2$. Let $f_1$ and $f_2$ be the neighbors of $s_1$ and $s_2$. Let $s_1'$ and $s_2'$ be the sticks of $e_1$ and $e_2$ contained in $f_1$ and $f_2$, respectively. W.l.o.g.~assume that $(f_1,f) \in E[H]$, which implies that stick $s_1'$ belongs to a scissor or to a twin contained in $f_1$. We study each of these two cases separately in the following.

\begin{description}
\item[\boldmath $s_1'$ belongs to a scissor contained in $f_1$:] In this case, $f_1$ has to be adjacent to $f$, due to $2$-planarity. Also, edge $e_1'$ corresponding to the stick of the scissor different from $s_1'$ has to end in $f_2$ and due to bipartiteness in a different vertex than the one $s_2'$ ends; refer to the red edge of Fig.\ref{fig:3-face-incoming-2}. It is not difficult to see that in this case Lemma~\ref{lem:2-face} holds. Hence, $h(f_2)=2$. This implies that $(f_1,f_2) \in E[H]$ and thus $(f_1,f) \notin E[H]$; a contradiction.

\item[\boldmath $s_1'$ belongs to a twin contained in $f_1$:] As opposed to the previous case, $f_1$ cannot be adjacent to $f$. Let $e_1'$ be the edge containing the stick of this twin that is different from $s_1'$. To cope with this case, we have to consider two subsases based on whether $f_2$ is adjacent to $f$ or not: %
\begin{inparaenum}[(i)]
\item \label{sc:adjacent} $f_2$ is adjacent to $f$, and
\item \label{sc:opposite} $f_2$ is not adjacent to $f$.
\end{inparaenum}

Consider first Case (\ref{sc:adjacent}), in which $f_2$ is adjacent to $f$. Then, $e_1$ has a middle-part in $f_2$ (bold-drawn in Figs.~\ref{fig:3-face-incoming-3}-\ref{fig:3-face-incoming-5}) and ends in one of the three faces adjacent to $f_2$ that is not identified with $f$; each of these cases is illustrated in Figs.~\ref{fig:3-face-incoming-3}, \ref{fig:3-face-incoming-4} and \ref{fig:3-face-incoming-5}. In each of these three case, $e_1'$ has to end in $f_2$ due to $2$-planarity. Additionally, its corresponding stick in $f_2$ forms a pseudo-scissor with $s_2'$. Now, it is not difficult to see that $f_2$ cannot contain more sticks in $f_2$, which implies that $h(f_2)=2$. Hence, $(f_1,f_2) \in E[H]$ and $(f_1,f) \notin E[H]$; a contradiction.

Consider now Case (\ref{sc:opposite}), in which $f_2$ is not adjacent to $f$. Recall that $f_1$ is not adjacent to $f$ as well. Let $f'$ be the face that is adjacent to $f$, $f_1$ and $f_2$; see e.g., Fig.\ref{fig:3-face-incoming-6}. Let also $f^*$ be the fourth face that is adjacent to $f'$ and is not identified with $f$, $f_1$ and $f_2$. Note that $f'$ contains  the two middle-parts of $e_1$ and $e_2$ (bold-drawn in Figs.~\ref{fig:3-face-incoming-6}-\ref{fig:3-face-incoming-14}). As above, we assume that $(f_1,f) \in E[H]$. Thus, $s_1'$ has to belong to a twin in $f_1$. For each of $f_1$ and $f_2$, there exist three different ``positions'' with respect to face $f$; we refer to them as \emph{left}, \emph{down} and \emph{right} (for example, Fig.\ref{fig:3-face-incoming-6} illustrates the case where $f_1$ is left and $f_2$ is down with respect to $f$). This gives rise to six different configurations for $f_1$ and $f_2$. Let $e_1'$ be the edge whose stick is defining the twin (besides $s_1'$) in $f_1$ that has to cross the shared edge of $f'$ and $f_1$ by the definition of a twin.

\begin{itemize}[-]
\item \emph{$f_1$ is left and $f_2$ is down}: Refer to Fig.~\ref{fig:3-face-incoming-6}.
Edge $e_1'$ cannot end in $f$, since we disallow homotopic edges. Also, edge $e_1'$ cannot end at a vertex of a different partition in $f'$ or in $f_2$ without crossing $e_2$, which is already involved in two crossings. Thus, edge $e_1'$ cannot exist and therefore $(f_1,f) \notin E[H]$; a contradiction.

\item \emph{$f_2$ is left and $f_1$ is down}: Refer to Figs.~\ref{fig:3-face-incoming-7-1} and~\ref{fig:3-face-incoming-7-2}.
Edge $e_1'$ cannot end in $f$ for the same reason as before. Also, edge $e_1'$ cannot end at a vertex of a different partition in $f_2$ without crossing $e_2$, which is already involved in two crossings. Thus, $e_1'$ has to end in either $f'$ or in $f^*$. First, consider the case where $e_1'$ ends in $f'$; see Fig.~\ref{fig:3-face-incoming-7-1}. In this case, $e_1$ and $e_1'$ prevent $f'$ to contain any other stick (due to $2$-planarity). This implies that $h(f')= 1$ and $(f_1,f') \in E[H]$. Hence, $(f_1,f) \notin E[H]$; a contradiction. Consider now the second case, in which $e_1'$ ends in $f^*$; see Fig.~\ref{fig:3-face-incoming-7-2}. In this case, the edge shared by $f_1$ and $f'$ has two crossings, which implies that no other edge is crossing this edges. To rule out this case, we apply a simple trick. While keeping the number of edges unchanged, we reroute $e_1'$ to end in $f'$ and then the previous case applies.

\item \emph{$f_1$ is right and $f_2$ is down}: Refer to Fig.~\ref{fig:3-face-incoming-8}.
Edge $e_1'$ cannot end in $f$, since we disallow homotopic edges. Also, edge $e_1'$ can end neither in $f'$ nor in $f^*$ without crossing edge $e_2$, which is already involved in two crossings. Hence, $e_1'$ ends in $f_2$ and in particular at a vertex of $f_2$ that is different from the one where $e_2$ ends. By Lemma\ref{lem:2-face}, $h(f_2) = 2$ holds. This implies that $(f_1,f_2) \in E[H]$. Hence, $(f_1,f) \notin E[H]$, a contradiction.

\item \emph{$f_2$ is right and $f_1$ is down}: Refer to Figs.~\ref{fig:3-face-incoming-9} and~\ref{fig:3-face-incoming-10}.
As before, edge $e_1'$ cannot end in $f$ and in $f'$. Thus, $e_1'$ has to end either in $f_2$ or in $f^*$. Consider first the case where $e_1'$ ends in $f_2$; see Fig.~\ref{fig:3-face-incoming-9}. In this case, $e_1'$ has to end in a vertex adjacent to the one edge $e_2$ is incident to. By Lemma\ref{lem:2-face}, $h(f_2) = 2$ holds. This implies that $(f_1,f_2) \in E[H]$. Hence, $(f_1,f) \notin E[H]$; a contradiction. Consider now the case where $e_1'$ ends in $f^*$; see Fig.~\ref{fig:3-face-incoming-10}. Since the edge shared by $f$ and $f'$ is crossed twice, there exists no additional edge starting in $f$ that potentially ends in $f'$. Furthermore, since the opposite edge (shared by $f'$ and $f_1$) is crossed twice as well, we can reroute $e_1$ to end in $f'$ instead (refer to the green-colored edge of Fig.~\ref{fig:3-face-incoming-10}). Since we eliminated the twin in $f_1$, it follows that $(f_1,f) \notin E[H]$ obviously holds. The contradiction is then obtained by following the augmentation of Case~(\ref{sc:adjacent}).

\item \emph{$f_1$ is left and $f_2$ is right}: Refer to Figs.~\ref{fig:3-face-incoming-11} and~\ref{fig:3-face-incoming-12}.
Following similar arguments as in our previous cases, we first show that edge $e_1'$ can end neither in $f$ and nor in $f'$. Hence, edge $e_1'$ has to end either in $f_2$ or in $f^*$. First consider the case, where $e_1'$ ends in $f_2$; see Fig.~\ref{fig:3-face-incoming-11}. In this case, the endvertex of $e_1'$ contained in $f_2$ is adjacent to the corresponding one of $e_2$ in $f_2$ (by bipartiteness). By Lemma\ref{lem:2-face}, $h(f_2) = 2$ holds. This implies that $(f_1,f_2) \in E[H]$. Hence, $(f_1,f) \notin E[H]$; a contradiction. Consider now the case where $e_1'$ ends in $f^*$; see Fig.~\ref{fig:3-face-incoming-12}. As in the second part of the previous case, we can reroute $e_1$ such that it ends in $f'$ instead of $f_1$ (refer to the green-colored edge of Fig.~\ref{fig:3-face-incoming-12}). Obviously, $(f_1,f) \notin E[H]$. The contradiction is then obtained by following again the augmentation of Case~(\ref{sc:adjacent}).

\item \emph{$f_2$ is left and $f_1$ is right}: Refer to Figs.~\ref{fig:3-face-incoming-13} and~\ref{fig:3-face-incoming-14}.
In this case, $e_1'$ ends either in $f'$ or in $f^*$. Consider the former case; see Fig.~\ref{fig:3-face-incoming-13}. In this case, by $2$-planarity face $f'$ contains exactly one stick, i.e., $h(f') = 1$. This implies that $(f_1,f') \in E[H]$. Hence, $(f_1,f) \notin E[H]$; a contradiction. For the later case, refer to Fig.~\ref{fig:3-face-incoming-14}. We proceed by rerouting $e_1'$ such that it ends in $f'$ (instead of ending in $f^*$). This reduces this case to the first one of our case analysis.
\end{itemize}
\end{description}

To complete the proof of this lemma, it remains to consider Case~C.\ref{mc:neither}. By definition of the dependency graph  $H$, $f$ has no outgoing edge in $H$.  We will now prove that in this particular case, $f$ has no incoming edge as well and thus the statement holds. Denote by $s_1$, $s_2$ and $s_3$ the three sticks of $f$ (recall that $h(f)=3$). Since $f$ contains no twin, at least two of these sticks have to cross in $f$. Since $f$ contains no scissor, the crossing sticks, say $s_1$ and $s_2$, must form a pseudo-scissor. First observe that if $s_1$ and $s_2$ cross the same boundary edge of $f$, then by Lemma~\ref{lem:2-face} $h(f)=2$; a contradiction. Hence, $s_1$ and $s_2$ either cross two adjacent boundary edges of $f$ or two opposite boundary edges of $f$. We consider each of these two case separately.

\begin{figure}[t]
	\centering
    \subcaptionbox{\label{fig:prop2pseudo-1}}{\includegraphics[width=0.2\textwidth,page=1]{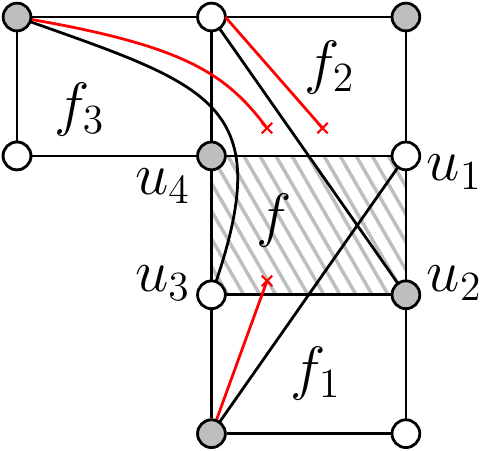}}
     \hfil
    \subcaptionbox{\label{fig:prop2pseudo-2}}{\includegraphics[width=0.25\textwidth,page=2]{prop2pseudo}}
    \caption{
    Different configurations used in the proof of Case~C.\ref{mc:neither} of Property~\ref{prp:3sticks}.}
    \label{fig:prop2pseudo}
 \end{figure}

\begin{description}
\item[\boldmath $s_1$ and $s_2$ cross two adjacent boundary edges of $f$:] Let $f = (u_1,u_2,u_3,u_4)$ and w.l.o.g.~let $s_1$ be incident to $u_1$ and $s_2$ be incident to $u_2$. Suppose first that $s_1$ and $s_2$  opposite boundary edges of $f$; see Fig.~\ref{fig:prop2pseudo-1}. This implies that $s_1$ crosses $(u_2,u_3)$ and $s_2$ crosses $(u_4,u_1)$. Stick $s_3$ has to be incident to either $u_3$ and cross $(u_4,u_1)$ or to $u_4$ and cross $(u_2,u_3)$, say w.l.o.g.~the former. Denote by $f_1$ the face where the edge $e_1$ that corresponds to stick $s_1$ ends, by $f_2$ the face where the edge $e_2$ that corresponds to stick $s_2$ ends and by $f_3$ the face where the edge $e_3$ that corresponds to stick $s_3$ ends. Note that $f_2 \neq f_3$, since by bipartiteness this would imply an additional crossing of $e_2$. Assume first that $(f_1,f) \in E[H]$. Since $e_1$ is already involved in two crossings, the stick of $e_1$ in $f_1$ has to be part of a twin. This means, however, that the other edge whose sticks form this twin would cross $(u_2,u_3)$ and enter $f$. Since $(u_4,u_1)$ is already crossed twice, it is easy to see that this edge can neither end in $f$ nor in $f_2$ (or any other face adjacent to $f$ that is not $f_1$). Thus this edge cannot exist and $(f_1,f) \notin E[H]$; a contradiction. Assume now that $(f_2,f) \in E[H]$. Since $e_2$ cannot be part of a scissor (due to $2$-planarity) and since $e_2$ cannot be part of a twin (because edge $(u_4,u_1)$ is already crossed twice), $f_2$ does not have $f$ as a neighbor in $H$, i.e., $(f_2,f) \notin E[H]$; a contradiction. Finally, assume that $(f_3,f) \in E[H]$. Since $e_3$ has two crossings, the stick of $e_3$ in $f_3$ has to be part of a twin. The second edge forming this twin in $f_3$, call it $e_3'$, has to cross the edge shared by $f_3$ and $f_2$. Since we disallow homotopic multiedges, $e_3'$ cannot end in $f_2$. Since $(u_4,u_1)$ is already crossed twice, $e_3'$ cannot end in $f$. Hence, this edge does not exist and therefore $(f_3,f) \notin E[H]$; a contradiction.

\item[\boldmath $s_1$ and $s_2$ cross two adjacent boundary edges of $f$:] Assume w.l.o.g.~that $s_1$ is incident to $u_3$ and $s_2$ is incident to $u_2$ and that $s_1$ crosses $(u_1,u_2)$ while $s_2$ crosses $(u_4,u_1)$; see Fig.~\ref{fig:prop2pseudo-2}. Due to $2$-planarity, $s_3$ has to start in $u_3$ and cross $(u_4,u_1)$. Due to bipartiteness, $f_3 \neq f_2$ holds. Using similar arguments are in the previous case, we can prove that $(f_2,f) \notin E[H]$ and  $(f_3,f) \notin E[H]$. Suppose now that  $(f_1,f) \in E[H]$. Then, edge $e_1$ is part of a twin in $f_1$. The second edge forming this twin in $f_1$, call it $e_1'$ has to cross $(u_1,u_2)$. Since we disallow homotopic multiedges, $e_1'$ cannot end in $f$. Since $(u_4,u_1)$ is already crossed twice, $e_1'$ cannot end in $f_2$. Thus $(f_1,f) \notin E[H]$; a contradiction.
\end{description}

\noindent Since we have led to a contradiction Cases~C.\ref{mc:scissor}-C.\ref{mc:neither}, the statement follows.
\end{proof}

\begin{prp}\label{prp:2sticks}
Every face $f$ of $G_p$ with $h(f)=2$ has at most two incoming edges in $H$.
\end{prp}
\begin{proof}
Since $f$ has two sticks, it is a neighbor of at most two faces.
\end{proof}

\noindent We are now ready to state the main theorem of this section.

\begin{theorem}\label{thm:2-upper}
A bipartite $n$-vertex $2$-planar multigraph has at most $3.5n - 7$ edges.
\end{theorem}
\begin{proof}
From Properties~\ref{prp:4sticks}, \ref{prp:3sticks} and \ref{prp:2sticks}, we can conclude that the number of faces that contain two sticks is at least as large as the number of faces that contain four sticks in $G'$. Therefore, the average number of sticks for a face is at most $3$. Since each edge of $G$ that does not belong to the planar structure of $G$ has two sticks and since $G_p$ has $2n-4$ edges and $n-2$ faces, graph $G$ has at most $(2n-4) + \frac{1}{2}\cdot3\cdot(n-2) = 3.5n- 7$~edges.

To complete the proof, we now consider the general case in which $G$ may contain $8$-sticks configurations. In particular, consider an $8$-stick configuration of graph $G$, as illustrated in Fig.~\ref{fig:special-case-1} where we have denoted by $a_1,\dots,a_{12}$ the involved vertices in clockwise order around it. We proceed by adding a vertex $v$ inside the central face $f'$ of the $8$-stick configuration and add so-called \emph{diagonal} edges $(a_2,v)$ and $(a_8,v)$. This implies that $v$ is in a different partition than $a_i$ for even $i$. We also replace edges $(a_3,a_{12})$ and $(a_6,a_9)$ by edges $(a_6,v)$ and $(a_{12},v)$. Furthermore, we add $(a_3,a_8)$ and $(a_2,a_9)$ without violating $2$-planarity. The result of this operation is illustrated in Fig.~\ref{fig:special-case-2}, where the extra edges are drawn dashed-green and vertex $v$ is drawn as a box. By applying the aforementioned augmentation for $8$-sticks configuration, we obtain a new graph $G'$. The planar structure $G_p'$ of $G'$ is the union of $G_p$ and of all the diagonal edges, and thus a quadrangulation. Further, $G'$ contains no $8$-stick configuration. Hence, $G'$ has at most $3.5n'- 7$ edges, where $n'=n+n_8$ is the number of vertices of $G'$ and $n_8$ is the number of $8$-stick configurations of $G$, i.e., the number of newly added vertices. Since $G'$ has $4n_8$ edges more than $G$, it follows that $G$ has at most $3.5n-0.5n_8-7<3.5n-7$ edges. This concludes the proof of this theorem.
\end{proof}

\subsection{Proof of Lemma~\ref*{lem:sticks}}
\label{subsec:sticksproof}

\noindent In this subsection, we give the detailed proof of Lemma~\ref{lem:sticks}, which was earlier omitted.

\rephrase{Lemma}{\ref{lem:sticks}}{\quadrangulation}
\begin{proof}
By Lemma~\ref{lem:connected}, $G_p$ is connected. Hence, each of its faces is connected, but it is not necessarily simple.
Since $G$ is bipartite, a face of $G_p$ cannot have length less than $4$. Suppose for a contradiction that there is a face $f=\{u_0,u_1,\ldots,u_{k-1}\}$ with length $k > 4$ in $G_p$.
If $f$ contains neither sticks nor middle-parts, then we can add at least an edge between two vertices of $f$, without violating bipartiteness and without crossing any edge of $G$, which contradicts the maximality of $G_p$. This gives rise to two main cases in our proof:
\begin{enumerate}[C.1:]
\item \label{mc:middle-parts} $f$ contains no sticks, but middle-parts,
\item \label{mc:sticks} $f$ contains at least one stick.
\end{enumerate}

We start with Case~C.\ref{mc:middle-parts}. Note that edges corresponding to middle-parts already have two crossings with edges of $f$ and so they cannot be involved in any other crossing.
We consider the following subcases: %
\begin{inparaenum}[(i)]
\item \label{c:covered} every middle-part contained in $f$ crosses edges $(u_{i-1},u_i)$ and $(u_i,u_{i+1})$, for some $0 \leq i \leq k-1$, and
\item \label{c:long} there exists at least one middle-part $m$ contained in $f$ crossing two edges $(u_{i-1},u_i)$ and $(u_{j},u_{j+1})$, for some $0 \leq i \neq j \leq k-1$.
\end{inparaenum}

We first consider Case~(\ref{c:covered}). In this case we say that $u_i$ is \emph{covered} by this middle-part.
If every vertex $u_i$ of $f$ is covered by some middle-part, then each $u_i$ is covered by exactly one middle-part, by $2$-planarity. In this case, we remove the edges corresponding to the middle-part covering $u_1$ and to the one covering $u_4$. We then add edge $(u_1,u_4)$ to $G_p$, drawing it inside $f$; also, we add $(u_1,w)$ to $G \setminus G_p$, where $w$ is the endvertex of the edge that corresponds to the middle-part covering $u_4$ and that does not belong to the same partition as $u_1$; we draw $(u_1,w)$ by first following $(u_1,u_4)$ and then the edge corresponding to the middle-part covering $u_4$ until $w$.
Otherwise, there exists a vertex $u_i$ of $f$ that is not covered. Then, we remove the edges corresponding to the middle-parts covering $u_{i+3}$, which are at most two by $2$-planarity. We add edge $(u_i,u_{i+3})$ to $G_p$, drawing it inside $f$; also, if we removed two edges, then we add $(u_i,w)$ to $G \setminus G_p$, where $w$ is an endvertex of one of the removed edges, drawing it as in the previous step.

We now proceed with Case~(\ref{c:long}). First observe that if there exists no other middle-part crossing $(u_{i-1},u_i)$ or $(u_{j},u_{j+1})$, then we can replace the edge corresponding to $m$  with one of the four possible edges connecting $u_{i-1},u_i,u_{j}$ and $u_{j+1}$; note that either $u_i$ and $u_j$ or $u_{i-1}$ and $u_{j+1}$ are not consecutive in $f$, since $f$ has length at least $6$.
In the following we will consider the cases where $j = i+1$ and $j \neq i+1$ separately.

\begin{description}
\item[\boldmath $j = i+1$:] Assume first that there exists an additional middle-part crossing both $(u_{i-1},u_i)$ and $(u_{j},u_{j+1})$.
This implies that we can add edge $(u_{i-1},u_{j+1})$ without crossing any edges by routing it along $m$, a contradiction to the maximality of $G$.
Assume now that there exists two additional middle-parts $m_1$ and $m_2$ crossing $(u_{i-1},u_i)$ and $(u_{i+1},u_{i+2})$, respectively. If $m_1$ covers $u_i$ and $m_2$ covers $u_{j}$, we can add edge $(u_{i-1},u_{j+1})$ as above and obtain a contradiction.
If $m_1$ crosses $(u_{i-1},u_i)$  but does not cover $u_i$ (that is, $m_1$ crosses an edge $(u_q,u_{q+1})$ with $q \neq i,j$) and $m_2$ covers $u_{j}$, we can remove $m_1$ and add $(u_{i-1},u_{j+1})$, a contradiction to the maximality of $G_p$. The case in which $m_1$ covers $u_i$ and $m_2$ does not cover $u_j$ is symmetric.
If neither $m_1$ covers $u_i$ nor $m_2$ covers $u_j$, we remove $m_1$ and $m_2$, add $(u_{i-1},u_{j+1})$ and add either $(u_{i-1},w)$ or $(u_i,w)$, where $w$ is the endvertex of the edge corresponding to $m$ in the face that shares edge $(u_{j},u_{j+1})$ with $f$. This ensures that if there already exists a copy of that edge, then this copy is not homotopic to the one we added.

\item[\boldmath $j \neq i+1$:] Suppose w.l.o.g.~that $j = i+k$ for some $ 2 \leq k \leq n-4$. For the first case we again assume that there exists an additional middle-part $m_1$ crossing both $(u_{i-1},u_i)$ and $(u_{j},u_{j+1})$. We proceed by removing both edges corresponding to $m$ and $m_1$ and adding the two possible edges between $u_{i-1},u_i,u_j$ and $u_{j+1}$ inside $f$ that do not violate bipartiteness. These two edges cross each other, but do not cross any other edge,  thus we can add one of them to $G_p$ and the other one to $G \setminus G_p$. Since both the removed middle-parts correspond to edges of $G \setminus G_p$, we again have a contradiction to the maximality of~$G_p$. As above, we consider now the case where there exist two middle-parts $m_1$ and $m_2$ that cross $(u_{i-1},u_{i})$ and $(u_{j},u_{j+1})$, respectively. Let $P_1$ and $P_2$ be the paths of $f$ consisting of edges $(u_{j+1},u_{j+2}),\dots,(u_{i-2},u_{i-1})$, and $(u_{i},u_{i+1}),\dots,(u_{j-1},u_{j})$, respectively.

\medskip
\emph{Suppose first that $u_i$ is in a different partition than $u_j$}. If neither $m_1$ nor $m_2$ cross an edge of $P_2$, we can simply add $(u_i,u_j)$, thus contradicting the maximality of $G$. If one of them crosses an edge of $P_2$, say $m_1$, we remove the edge corresponding to $m_1$ and add $(u_i,u_j)$ to $G_p$, a contradiction to the maximality of $G$. Symmetrically, if both $m_1$ and $m_2$ cross an edge of $P_2$, then we can simply add $(u_{i-1},u_{j+1})$ and obtain again a contradiction to the maximality of $G$. It follows that either $m_1$ crosses an edge of $P_1$ and $m_2$ crosses an edge of $P_2$ or vice versa; assume w.l.o.g.~the former. In this case we remove $m_1$ and add edge $(u_{i-1},u_{j+1})$ to $G_p$, which leads to a contradiction the maximality of~$G_p$.

\medskip
\emph{Suppose now that $u_i$ is in the same partition as $u_j$}. If $m_1$ does not cross an edge of $P_2$ and $m_2$ does not cross an edge of $P_1$, then we can remove the edge corresponding to $m$ and add edge $(u_i,u_{j+1})$ to $G_p$ and obtain a contradiction. If $m_1$ crosses an edge of $P_2$ and $m_2$ crosses an edge of $P_1$, we can symmetrically add edge $(u_{i-1},u_{j})$ by first removing the edge corresponding to $m$.
We now assume that both $m_1$ and $m_2$ cross and edge of $P_1$ or an edge of $P_2$, say w.l.o.g.~the former. We proceed by removing $m$ and $m_1$ and adding $(u_{i-1},u_{j})$ to $G_p$ and either $(u_{i},w)$ or $(u_{i-1},w)$ to $G \setminus G_p$, where $w$ is the endvertex of the edge corresponding to $m$ in the face that shares edge $(u_{j},u{j+1}$ with $f$. Our choice depends on the partition in which $w$ belongs to. In both cases, we obtain a contradiction to the maximality of~$G_p$.
\end{description}

We continue our description with Case~C.\ref{mc:sticks} of our case analysis. Here, we will first show that no two two sticks contained in $f$ cross (if any). To prove this claim, we assume to the contrary that $f$ contains two sticks $s$ and $s'$ that cross. We consider the following subcases: %
\begin{inparaenum}[(i)]
\item \label{c:two-long} both $s$ and $s'$ are long,
\item \label{c:two-short} both $s$ and $s'$~are short, and
\item \label{c:short-long} $s$ is short and $s'$ is long.
\end{inparaenum}

We start our description with Case~(\ref{c:two-long}). Denote by $u_i$ and $u_j$ the vertices of $f$, stick $s$ and $s'$ are incident to; see Fig.\ref{fig:2longs-1}. Let $(u_{i'},u_{i'+1})$ and $(u_{j'},u_{j'+1})$ be the edges of $f$ that are crossed by sticks $s$ and $s'$. By $2$-planarity, each of $(u_{i'},u_{i'+1})$ and $(u_{j'},u_{j'+1})$ has at most one additional crossing. We only consider the case where both additional crossings exist, since all other cases are subsumed by it. Let $e$ and $e'$ be the edges of $G$ defining these crossings. First, we consider the case where $u_i$ and $u_j$ belong to different partitions of $G$. We proceed by removing the edges corresponding to $s$ and $s'$ and edges $e$ and $e'$, which allows us to connect each of $u_i$ and $u_j$ with the two vertices in $\{u_{i'},u_{i'+1},u_{j'},u_{j'+1}\}$ that do not violate bipartiteness; see Fig.\ref{fig:2longs-2}. Since the newly introduced edges are contained in $f$, we have not introduced homotopic edges. However, at least one of the newly introduced edges can be added to $G_p$, contradicting its maximality. We now consider the case where $u_i$ and $u_j$ belong to the same partition $\mathcal{W}$. Assume first that either $u_{i'}$ or $u_{j'+1}$, say w.l.o.g.~$u_{i'}$, does not belong to $\mathcal{W}$. In this case, we proceed by removing the edge corresponding to $s$ from $G$, which allows us to connect $u_j$ with $u_{i'}$. This leads to a contradiction maximality of $G_p$, as the newly introduced edge can be added to it. To complete our case analysis, assume now that both $u_{i'}$ and $u_{j'+1}$ belong to $\mathcal{W}$. We proceed by removing the edges corresponding to $s$ and $s'$ from $G$, which allows us to add edges $(u_i, u_{j'})$ and $(u_j, u_{i'+1})$, which leads again to a contradiction the maximality of $G_p$; see Fig.\ref{fig:2longs-3}.

\begin{figure}[t]
   \centering
   \subcaptionbox{\label{fig:2longs-1}}{\includegraphics[width=0.19\textwidth,page=1]{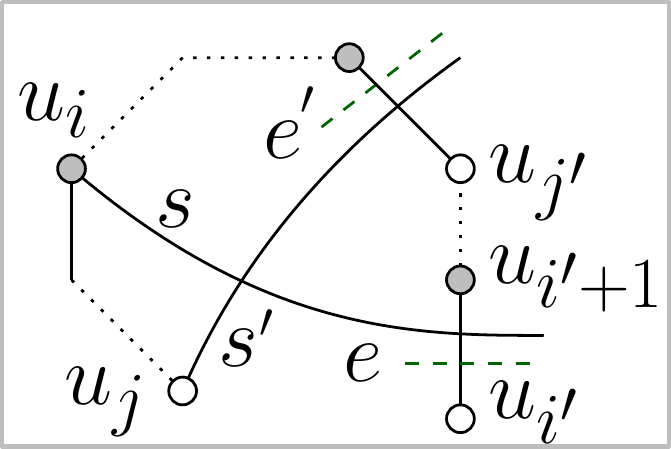}}
   \hfil
   \subcaptionbox{\label{fig:2longs-2}}{\includegraphics[width=0.19\textwidth,page=2]{2longs}}
   \hfil
   \subcaptionbox{\label{fig:2longs-3}}{\includegraphics[width=0.19\textwidth,page=3]{2longs}}
   \caption{%
   Different configurations in Case~C.\ref{mc:sticks}.(\ref{c:two-long}) of the proof of Lemma~\ref{lem:sticks}.}
   \label{fig:stick-types-lemma}
\end{figure}

We continue with Case~(\ref{c:two-short}). Let $u_i$ be the vertex of $f$ stick $s$ is incident to and assume w.l.o.g.~that $s$ crosses $(u_{i+1},u_{i+2})$ of $f$. Since both $s$ and $s'$ are short and cross, it follows that $s'$ can be incident to $u_{i+1}$, $u_{i+2}$ or $u_{i+3}$. In the following, we consider each of these~cases.

\begin{figure}[p]
   \centering
   \subcaptionbox{\label{fig:consecutiveshort-1}}{\includegraphics[width=0.24\textwidth,page=1]{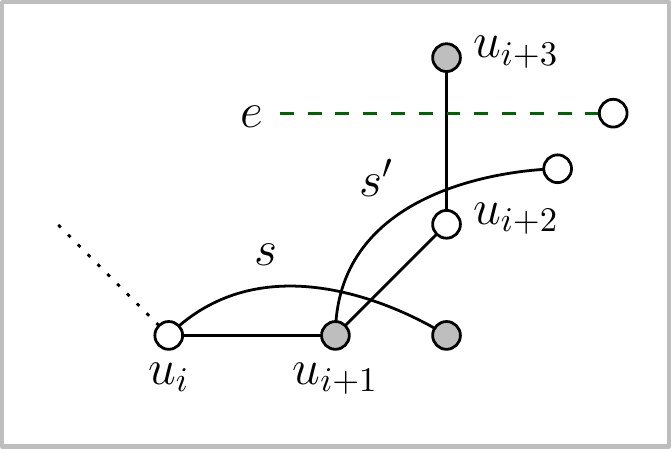}}
   \hfil
   \subcaptionbox{\label{fig:consecutiveshort-2}}{\includegraphics[width=0.24\textwidth,page=2]{consecutiveshort}}
   \hfil
   \subcaptionbox{\label{fig:oppositeshort-1}}{\includegraphics[width=0.24\textwidth,page=1]{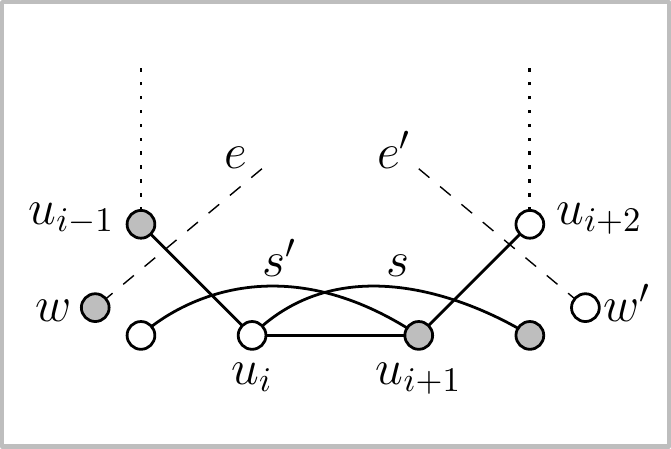}}
   \hfil
   \subcaptionbox{\label{fig:oppositeshort-2}}{\includegraphics[width=0.24\textwidth,page=2]{oppositeshort}}
   \hfil
   \subcaptionbox{\label{fig:oppositeshort-3}}{\includegraphics[width=0.24\textwidth,page=3]{oppositeshort}}
   \hfil
   \subcaptionbox{\label{fig:oppositeshort-4}}{\includegraphics[width=0.24\textwidth,page=4]{oppositeshort}}
   \hfil
   \subcaptionbox{\label{fig:missing-case-1-1}}{\includegraphics[width=0.24\textwidth,page=1]{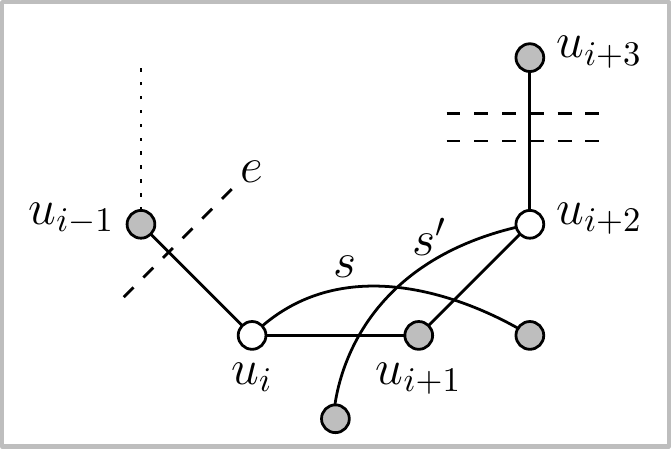}}
   \hfil
   \subcaptionbox{\label{fig:missing-case-1-2}}{\includegraphics[width=0.24\textwidth,page=2]{missing-case}}
   \hfil
   \subcaptionbox{\label{fig:missing-case-1-3}}{\includegraphics[width=0.24\textwidth,page=3]{missing-case}}
   \hfil
   \subcaptionbox{\label{fig:missing-case-1-4}}{\includegraphics[width=0.24\textwidth,page=4]{missing-case}}
   \hfil
   \subcaptionbox{\label{fig:missing-case-1-5}}{\includegraphics[width=0.24\textwidth,page=5]{missing-case}}
   \hfil
   \subcaptionbox{\label{fig:missing-case-2-1}}{\includegraphics[width=0.24\textwidth,page=6]{missing-case}}
   \caption{%
   Different configurations in Case~C.\ref{mc:sticks}.(\ref{c:two-short}) of the proof of Lemma~\ref{lem:sticks}.}
   \label{fig:short-lemma}
\end{figure}

\begin{description}
\item[\boldmath $s'$ is incident to $u_{i+1}$:] In this subcase, stick $s'$ crosses either $(u_{i+2},u_{i+3})$ or $(u_{i-1},u_i)$; see Figs.~\ref{fig:consecutiveshort-1} and~\ref{fig:oppositeshort-1}. Consider first the case where $s'$ crosses $(u_{i+2},u_{i+3})$. If we could add edge $(u_i,u_{i+3})$ without violating $2$-planarity, then this would lead to a contradiction the maximal density of $G$.  Since both edges corresponding to $s$ and $s'$ have two crossings each, the only edge that can prevent the addition of edge $(u_i,u_{i+3})$ is an edge, call it $e$, crossing $(u_{i+2},u_{i+3})$ of~$f$; see Fig.\ref{fig:consecutiveshort-1}. In particular, we could draw $(u_i,u_{i+3})$ starting from $u_i$, following the curve of $s$ until its intersection with $s'$, then follow the curve of $s'$ until its intersection with $(u_{i+2},u_{i+3})$ and finally follow the curve of this edge until reaching $u_{i+3}$. Note that there is only one such edge, because edge $(u_{i+2},u_{i+3})$ is also crossed by the edge corresponding to $s'$. We proceed by removing edge $e$ from $G$, which allows us to add edge $(u_i,u_{i+3})$; see Fig.~\ref{fig:consecutiveshort-2}. The obtained graph has the same number of edges as $G$, but it has a larger planar structure; a contradiction.

Consider now the case where $s'$ crosses $(u_{i-1},u_i)$; see, e.g.,~Fig.\ref{fig:oppositeshort-1}. In this case, if we could add edge $(u_{i-1},u_{i+2})$ without violating $2$-planarity, then this would lead to a contradiction the maximal density of $G$. Since each of the edges corresponding to $s$ and $s'$ has two crossings, the only edges that can prevent the addition of edge $(u_{i-1},u_{i+2})$ is either an edge $e$ crossing $(u_{i-1},u_i)$ of $f$ or an edge $e'$ crossing $(u_{i+1},u_{i+2})$ (or both of them). We only consider the case where both $e$ and $e'$ exist, because all other cases are subsumed by it. Let $w$ and $w'$ be the endpoints of $e$ and $e'$, respectively, that one can reach following these edges from their crossing points with the boundary of $f$, if one moves outside $f$.
For the case where $w$ and $w'$ belong to the same partition (say w.l.o.g.\ in the same partition as $u_i$), we proceed by deleting edges $e$ and $e'$ from $G$, which allows us add edges $(u_{i-1},u_{i+2})$ and $(u_{i-1},w')$; see Fig.\ref{fig:oppositeshort-2}. The same transformation can be applied when $w$ and $w'$ are in different partitions such that $w$ is in the same partition as $u_{i-1}$ and $w'$ is in the same partition as $u_i$. In both cases, the suggested transformation leads to a larger planar structure and is therefore a contradiction to $G_p$. Finally, it remains to consider the case where $w$ and $w'$ are in different partitions, such that $w$ is in the same partition as $u_i$ and $w'$ is in the same partition as $u_{i-1}$; see Fig.\ref{fig:oppositeshort-3}. Assume that $e=(w,z)$. In this case, we proceed by deleting edges $e$, $e'$ and the edge corresponding to $'$ from $G$. This allows us to add edges $(u_{i-1},u_{i+2})$, $(u_i,w')$ and $(z,u_{i+2})$; see Fig.\ref{fig:oppositeshort-4}. Note that this transformation is not possible, only if $z$ is $u_{i+3}$ of $f$. In this case, we do not add $(z,u_{i+2})$ to $G$ (which already exists), but we keep edge $e$. Since $e'$ is removed, $2$-planarity is maintained. Both cases yield a contradiction to the maximality~of~$G_p$.

\item[\boldmath $s'$ is incident to $u_{i+2}$:] Since $s$ and $s'$ cross, stick $s'$ crosses edge $(u_i,u_{i+1})$ of $f$. We claim that we can assume w.l.o.g.\ that each of $(u_{i-1},u_i)$ and $(u_{i+2},u_{i+3})$ of $f$ has two crossings. Suppose for a contradiction that, e.g., $(u_{i-1},u_i)$ has (at most) one crossing; see Fig.\ref{fig:missing-case-1-1}. We proceed by removing $(u_{i-1},u_i)$ from $G$, which allows us to add edge $(u_{i-1},u_{i+2})$ to $G$; see Fig.\ref{fig:missing-case-1-2}. This clearly leads to a contradiction the maximality of $G_p$. So, we can indeed assume w.l.o.g.\ that each of $(u_{i-1},u_i)$ and $(u_{i+2},u_{i+3})$ of $f$ has two crossings; see Fig.\ref{fig:missing-case-1-3}. Let $e_1$ and $e_1'$ be the edges crossing $(u_{i-1},u_i)$; let also $e_2$ and $e_2'$ be the edges crossing $(u_{i+2},u_{i+3})$. By $2$-planarity, we may assume that either one of $e_1$ and $e_1'$ does not end to $u_{i+3}$, or one of $e_2$ and $e_2'$ does not end to $u_{i-1}$. Assume w.l.o.g.\ that $e_1=(w,w')$ does not edge at $u_{i+3}$. We proceed by removing both  $e_1$ and $e_1'$ from $G$, which allows us to add edge $(u_{i-1},u_{i+2})$ and either the edge $(w,u_{i+2})$ or the edge $(w',u_{i+2})$ depending on whether $w$ or $w'$ does not belong to the same partition as $u_{i+2}$, respectively; see Figs.\ref{fig:missing-case-1-4}-\ref{fig:missing-case-1-5}. Again, we obtain a contradiction to the maximality of $G_p$.

\item[\boldmath $s'$ is incident to $u_{i+3}$] Since $s$ and $s'$ cross, stick $s'$ crosses edge $(u_{i+1},u_{i+2})$ of $f$; see Fig.\ref{fig:missing-case-2-1}. In this case, we can add edge $(u_i,u_{i+3})$, which contradicts the maximality of $G_p$.
\end{description}

\begin{figure}[p]
   \centering
   \subcaptionbox{\label{fig:shortlongout-1}}{\includegraphics[width=0.24\textwidth,page=1]{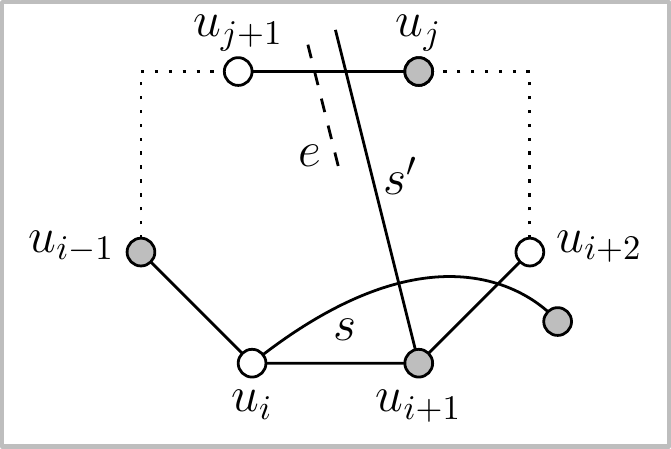}}
   \hfil
   \subcaptionbox{\label{fig:shortlongout-2}}{\includegraphics[width=0.24\textwidth,page=2]{shortlongout}}
   \hfil
   \subcaptionbox{\label{fig:shortlongout-3}}{\includegraphics[width=0.24\textwidth,page=3]{shortlongout}}
   \hfil
   \subcaptionbox{\label{fig:shortlongout-4}}{\includegraphics[width=0.24\textwidth,page=4]{shortlongout}}
   \hfil
   \subcaptionbox{\label{fig:shortlonginc-1}}{\includegraphics[width=0.24\textwidth,page=1]{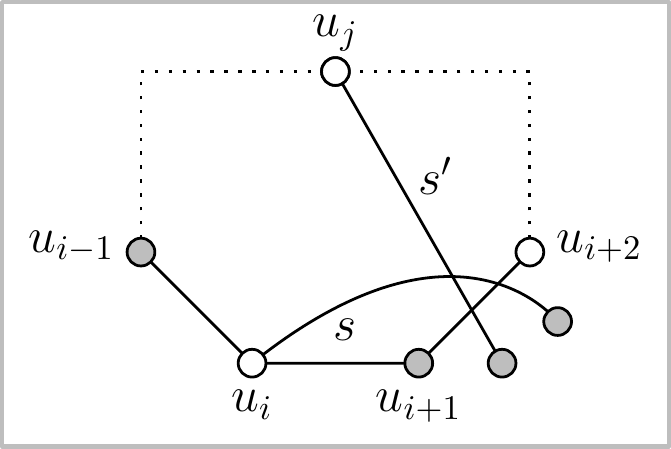}}
   \hfil
   \subcaptionbox{\label{fig:shortlonginc-2}}{\includegraphics[width=0.24\textwidth,page=2]{shortlonginc}}
   \hfil
   \subcaptionbox{\label{fig:shortlonginc-3}}{\includegraphics[width=0.24\textwidth,page=3]{shortlonginc}}
   \hfil
   \subcaptionbox{\label{fig:shortlonginc-4}}{\includegraphics[width=0.24\textwidth,page=4]{shortlonginc}}
   \hfil
   \subcaptionbox{\label{fig:shortlonginc-5}}{\includegraphics[width=0.24\textwidth,page=5]{shortlonginc}}
   \hfil
   \subcaptionbox{\label{fig:shortlonginc-6}}{\includegraphics[width=0.24\textwidth,page=6]{shortlonginc}}
   \hfil
   \subcaptionbox{\label{fig:shortlonginc-7}}{\includegraphics[width=0.24\textwidth,page=7]{shortlonginc}}
   \hfil
   \subcaptionbox{\label{fig:shortlonginc-8}}{\includegraphics[width=0.24\textwidth,page=8]{shortlonginc}}
   \caption{%
   Different configurations in Case~C.\ref{mc:sticks}.(\ref{c:short-long}) of the proof of Lemma~\ref{lem:sticks}.}
   \label{fig:long-and-short-lemma}
\end{figure}

Finally, we consider Case~(\ref{c:short-long}). Since $s$ is a short stick, we may assume as in the previous case that $u_i$ is the vertex of $f$ stick $s$ is incident to and that $s$ crosses $(u_{i+1},u_{i+2})$ of $f$. We distinguish two cases based on whether (long) stick $s'$ is incident to vertex $u_{i+1}$ of $f$ or not.

\begin{description}
\item[\boldmath $s'$ is incident to $u_{i+1}$:] Assume w.l.o.g.\ that stick $s'$ crosses the boundary edge  $(u_j,u_{j+1})$ of $f$. Since $s'$ is a long stick, it follows that $u_j \neq u_{i-1}$ and $u_{j+1} \neq u_{i+2}$; see Fig.\ref{fig:shortlongout-1}. By $2$-planarity, at most one more edge, call it $e$, can cross $(u_j,u_{j+1})$. Also, assume w.l.o.g.\ that $u_j$ is in the partition as $u_{i+1}$; the case where $u_j$ and $u_{i+1}$ are in different partitions is symmetric. If the first crossing along edge $(u_j,u_{j+1})$ when moving from vertex $u_j$ towards vertex $u_{j+1}$ is the one of stick $s$, we proceed by removing the edge corresponding to $s$ from $G$, which allows us to add edge $(u_i,u_j)$; see Fig.\ref{fig:shortlongout-2}. Otherwise (i.e., the first crossing along edge $(u_j,u_{j+1})$ when moving from vertex $u_j$ towards vertex $u_{j+1}$ is the one of edge $e$; see Fig.\ref{fig:shortlongout-3}), we proceed by removing the edge corresponding to $s$ from $G$, which now allows us to add edge $(u_i,u_{j+1})$; see Fig.\ref{fig:shortlongout-4}. Both cases form clearly contradictions to the maximality of $G_p$.

\item[\boldmath $s'$ is not incident to $u_{i+1}$:] Let w.l.o.g.\ $u_j$ be the vertex of $f$ stick $s'$ is incident to. In this case, $s'$ crosses $s$ and either $(u_i,u_{i+1})$ or $(u_{i+1},u_{i+2})$ of $f$. First, assume that $s'$ crosses $(u_{i+1},u_{i+2})$. If $u_i$ and $u_j$ are in the same partition (see Fig.\ref{fig:shortlonginc-1}), we proceed by removing the edge corresponding to $s$ from $G$, which allows us to add edge $(u_j,u_{i+1})$; see Fig.\ref{fig:shortlonginc-2}. Otherwise (i.e., $u_i$ and $u_j$ are in different partitions; see Fig.\ref{fig:shortlonginc-3}), we proceed by removing the edge corresponding to $s$ from $G$, which allows us to add edge $(u_i,u_j)$; see Fig.\ref{fig:shortlonginc-4}. Both cases form clearly contradictions to the maximality of $G_p$. To complete the proof, consider now the case, where $s'$ crosses $(u_i,u_{i+1})$ of $f$. If $u_i$ and $u_j$ are in the same partition (see Fig.\ref{fig:shortlonginc-5}), we proceed by removing the edge corresponding to $s$ and the second edge that might cross edge $(u_i,u_{i+1})$ of $f$ from $G$, which allows us to add to $G$ the edge $(u_j,u_{i+1})$ and the edge from $u_j$ towards the endpoint of the edge corresponding to $s$ that is different from $u_i$; see Fig.\ref{fig:shortlonginc-6}. Otherwise (i.e., $u_i$ and $u_j$ are in different partitions; see Fig.\ref{fig:shortlonginc-7}), we proceed by removing the edge corresponding to $s'$ from $G$, which allows us to add edge $(u_i,u_j)$; see Fig.\ref{fig:shortlonginc-8}. Both cases form clearly contradictions to the maximality of $G_p$.
\end{description}

Cases~(\ref{c:two-long})-(\ref{c:short-long}) imply that our initial claim that, no two sticks contained in $f$ cross, holds. We continue our case analysis by considering two more cases, namely:%
\begin{inparaenum}[(i)]\setcounter{enumi}{3}
\item \label{c:planar-short} all sticks contained in $f$ are short,
\item \label{c:planar-long} there exist at least one long stick contained in $f$.
\end{inparaenum}
In both cases, we show that it possible to add an edge in the interior of $f$, contradicting the maximality of $G_p$.

First consider  Case~(\ref{c:planar-short}). Let $s_1$ be a stick contained in $f$ incident to vertex $u_{i-1}$ and assume w.l.o.g.\ that $s_1$ crosses edge $(u_i,u_{i+1})$ of $f$. If there is no or only one (short) stick from vertex $u_{i+1}$ crossing edge $(u_{i+2},u_{i+3})$ as in Fig.\ref{fig:shorthalfplanar-1}, then we could remove it, and this would allow us to add edge $(u_{i-1},u_{i+2})$ in $G$ as in Fig.\ref{fig:shorthalfplanar-2}, thus contradicting the maximality of $G_p$. Hence, we may assume w.l.o.g.\ that there exist two (short) sticks, call them $s_2$ and $s_2'$, from vertex $u_{i+1}$ crossing edge $(u_{i+2},u_{i+3})$. Symmetrically, we also assume that, except from $s_1$, there is a second sticks incident to $u_{i-1}$, call it $s_1'$; see Fig.\ref{fig:shorthalfplanar-3}. In this case, we proceed by removing from $G$ both edges corresponding to sticks $s_2$ and $s_2'$. This allows us to prove add edge $(u_{i-1},u_{i+2})$ in $G$, as well as the edge from $u_{i-1}$ to the endvertex of the edge corresponding to one of $s_2$ or $s_2'$ that is different from $u_{i+1}$. This again leads to a contradiction the maximality of $G_p$.

\begin{figure}[t!]
   \centering
   \subcaptionbox{\label{fig:shorthalfplanar-1}{}}
   {\includegraphics[width=0.24\textwidth,page=1]{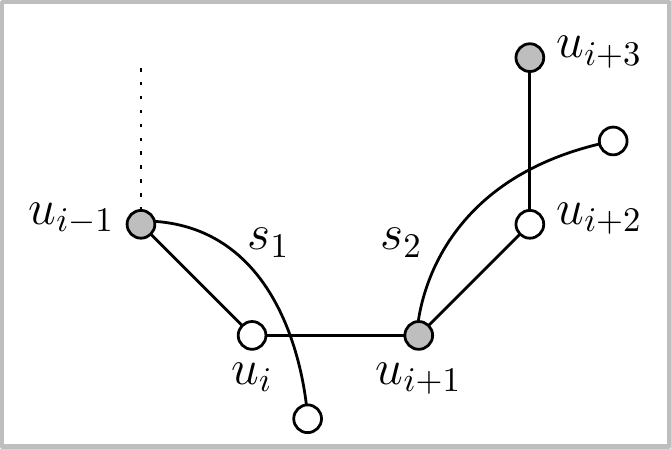}}
   \hfil
   \subcaptionbox{\label{fig:shorthalfplanar-2}{}}
   {\includegraphics[width=0.24\textwidth,page=2]{shorthalfplanar}}
   \hfil
   \subcaptionbox{\label{fig:shorthalfplanar-3}{}}
   {\includegraphics[width=0.24\textwidth,page=3]{shorthalfplanar}}
   \hfil
   \subcaptionbox{\label{fig:shorthalfplanar-4}{}}
   {\includegraphics[width=0.24\textwidth,page=4]{shorthalfplanar}}

   \subcaptionbox{\label{fig:longhalfplanar-1}{}}
   {\includegraphics[width=0.24\textwidth,page=1]{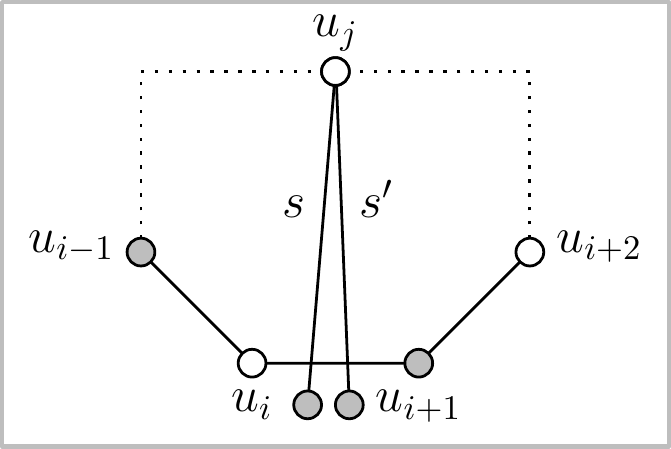}}
   \hfil
   \subcaptionbox{\label{fig:longhalfplanar-2}{}}
   {\includegraphics[width=0.24\textwidth,page=2]{longhalfplanar}}
   \hfil
   \subcaptionbox{\label{fig:longhalfplanar-3}{}}
   {\includegraphics[width=0.24\textwidth,page=3]{longhalfplanar}}
   \hfil
   \subcaptionbox{\label{fig:longhalfplanar-4}{}}
   {\includegraphics[width=0.24\textwidth,page=4]{longhalfplanar}}
   \caption{%
   Different configurations in (a)-(d) Case~C.\ref{mc:sticks}.(\ref{c:planar-short}), and (e)-(h) Case~C.\ref{mc:sticks}.(\ref{c:planar-long}) of the proof of Lemma~\ref{lem:sticks}.}
   \label{fig:shorthalfplanar}
\end{figure}

Now consider Case~(\ref{c:planar-long}). Let $s$ be a long stick contained in $f$. Assume w.l.o.g.\ that $s$ is incident to vertex $u_{j}$ of $f$ and that it crosses edge $(u_i,u_{i+1})$ of $f$, such that $u_{j}$ and $u_{i+1}$ are in the same partition. If edge $(u_i,u_{i+1})$ is only crossed by $s$, then we could add edge $(u_j,u_{i+1})$ in $G$, thus contradicting the optimality of $G$. Hence, we may assume wl.o.g.\ that $(u_j,u_{i+1})$ is crossed twice. First, assume that the second crossing along $(u_j,u_{i+1})$ is due to a stick $s'$ incident to $u_j$; see Fig.\ref{fig:longhalfplanar-2}. In this case, we could again add edge $(u_j,u_{i+1})$ in $G$, thus contradicting the optimality of $G$; see Fig.\ref{fig:longhalfplanar-2}. Assume now that the second crossing along $(u_j,u_{i+1})$ is due to a stick $s'$ of $f$ different from $u_j$. Let $u_k$ be the vertex of $f$ stick $s'$ is incident to; see Fig.\ref{fig:longhalfplanar-3}. In this case, we proceed by removing both sticks $s$ and $s'$ from $G$. This allows us to connect each of $u_j$ and $u_k$ with one of two vertices in $\{u_{i},u_{i+1}\}$ that do not violate bipartiteness; see Fig.\ref{fig:longhalfplanar-4}. To complete our case analysis, it remains to consider the case, where the second crossing of $(u_i,u_{i+1})$ is due to a middle-part contained in $f$. But in this case, we can simply remove it and this will allow us to replace it by the edge $(u_j,u_{i+1})$, thus contradicting the optimality of $G$.

From our case analysis, it follows that there exists a maximal dense $2$-planar bipartite graph $G$ with a maximal planar subgraph $G_p$ that consists of only quadrangular faces.
\end{proof}

\section{Two Applications of Theorem~\ref*{thm:2-upper}}
\label{sec:side-results}
In this section, we slightly improve the best known general lower bound on the number of crossings of a graph $G$, when $G$ is bipartite. Our proof is an adjustment of corresponding proofs for general (i.e., non bipartite) graphs; see e.g.,~\cite{Ackerman09,ajtal82,Leighton:1983:CIV:2304,PachRTT06}. Note that Zarankiewicz~\cite{Zarankiewicz54} back in 1954 posed his well-known conjecture about the exact crossing number of complete bipartite graphs. Here, we relax completeness. We start with the following lemma establishing a lower bound on the  crossing number $cr(G)$ of a bipartite graph $G$. 

\begin{lemma}
Let $G$ be a simple bipartite graph with $n \geq 3$ vertices and $m$ edges. Then, the crossing number $cr(G)$ satisfies the following:
\[
cr(G) \geq 3m - \frac{17}{2}n + 19
\]
\label{lem:crossings-bound}
\end{lemma}
\begin{proof}
The statements clearly holds when $m \leq 2n -4$. Hence, we may assume w.l.o.g.~that $m > 2n - 4$. It follows from~\cite{CzapPS16} that if $m > 3n - 8$, then $G$ has an edge that is crossed by at least two other edges. Also, by Theorem~\ref{thm:2-upper} we know that if $m > \frac{7}{2}n - 7$, then $G$ has an edge that is crossed by at least three other edges. We obtain by induction on the number of edges of $G$ that the crossing number $cr(G)$ is at least:
\[
cr(G) \geq (m - (2n-4) ) + (m - (3n - 8) ) + (m - (\frac{7}{2}n - 7)) = 3m - \frac{17}{2}n + 19
\]
\end{proof}

\begin{theorem}
Let $G$ be a simple bipartite graph with $n$ vertices and $m$ edges, where $m \geq \frac{17}{4}n$. Then, the crossing number $cr(G)$ satisfies the following:
\[
cr(G) \geq \frac{16}{289}\cdot\frac{m^3}{n^2}
\]
\label{thm:crossing-lemma}
\end{theorem}
\begin{proof}
Assume that $G$ admits a drawing on the plane with $cr(G)$ crossings and let $p=\frac{17n}{4m} \leq 1$. Choose independently every vertex of $G$ with probability $p$, and denote by $G_p$ the graph induced by the vertices chosen in $G_p$. Let also $n_p$, $m_p$ and $c_p$ be the random variables corresponding to the number of vertices, of edges and of crossings of $G_p$. Taking expectations on the relationship $c_p \geq 3m_p - \frac{17}{2}n_p + 19$, which holds by Lemma~\ref{lem:crossings-bound}, we obtain:
\[
p^4cr(G) \geq 3p^2m - \frac{17}{2}np ~~\Rightarrow~~ cr(G) \geq \frac{3m}{p^2} - \frac{17n}{2p^3}
\]
The proof of the theorem follows by plugging $p=\frac{17n}{4m}$ (which is at most $1$ by our assumption) to the inequality above.
\end{proof}

\begin{theorem}
Let $G$ be a simple bipartite $k$-planar graph with $n$ vertices and $m$ edges, for some $k \geq 1$. Then:
\[
m \leq \frac{17}{8}\sqrt{2k}n \approx 3.005 \sqrt{k}n
\]
\label{thm:general-bound}
\end{theorem}
\begin{proof}
For $k = 1$ and $k = 2$, the bounds of this theorem are weaker than the corresponding ones of \cite{CzapPS16} and of  Theorem~\ref{thm:2-upper}, respectively. So, we may assume w.l.o.g.~that $k > 2$. We may also assume that $m \geq \frac{17}{4}n$, as otherwise there is nothing to prove. Combining the fact that $G$ is $k$-planar with the bound of Theorem~\ref{thm:crossing-lemma} we obtain:
\[
\frac{16}{289}\cdot\frac{m^3}{n^2} \leq cr(G) \leq \frac{1}{2} mk 
\]
which implies:
\[
m \leq \frac{17}{8}\sqrt{2k}n \approx 3.005 \sqrt{k}n
\]
\end{proof}

\section{Conclusions and Open Problems}
\label{sec:conclusions}
In this paper, we studied beyond-planarity for bipartite graphs, focusing on Tur\'an-type problems. We proved bounds for the edge density that are tight up to additive constants for some of the most important classes. We conclude by listing several further open problems.

\begin{enumerate}
\item What is the maximum density of bipartite $k$-planar graphs with $k = 3,4,\ldots$? Note that tight such bounds will further improve the leading constant of the Crossing Lemma for bipartite graphs. Note also that our lower bound example for bipartite $2$-planar graphs can be extended to a dense bipartite $3$-planar graph with $4n - O(1)$ edges; see Fig.~\ref{fig:max-3-planar}. Of interest is also to study density~bounds for other classes of bipartite nearly-planar graphs, e.g., of bipartite quasi-planar graphs.

\item It is interesting to compare for a fixed class of nearly-planar graphs the ratio of the maximum density of general over bipartite graphs. Using our results, we observe that this ratio for large $n$ approaches $\frac{3n}{2n} = 1.5$ for planar graphs, $\frac{4n}{3n} \approx 1.33$ for $1$-planar graphs, $\frac{5n}{3.5n} \approx 1.43$ for $2$-planar graphs and at most $\frac{5.5n}{4n} \approx 1.37$ for $3$-planar graphs. For fan-planar graphs, the corresponding ratio is $\frac{5n}{4n} = 1.2$. The ratio varies and there is room for speculation on how it develops for $k$-planar graphs for increasing $k$. 
Note that, as long as the class is closed under subgraphs, it cannot be more than $2$, since any graph with $n$ vertices and $m$ edges has a bipartite subgraph with at least $\frac{m}{2}$ edges~\cite{Erdos65}.

\item  Another research direction, which has been intensively considered for several classes of nearly-planar graphs, are geometric graphs, where edges are represented as straight-line segments; see, e.g.,~\cite{DBLP:journals/ipl/Didimo13}. Note that we mostly use straight-line drawings for the lower bounds, but allow the more general topological graphs when proving upper bounds. 

\item Optimal $1$-, $2$- and $3$-planar graphs allow for complete characterizations~\cite{DBLP:conf/gd/Bekos0R16,DBLP:journals/siamdm/Suzuki10}. So, another reasonable question to ask is whether the same holds for the corresponding optimal bipartite graphs. In particular, for optimal $1$-planar graphs there is also an efficient recognition algorithm~\cite{Brandenburg16a}. Can such an algorithm be obtained for a corresponding bipartite class of nearly-planar graphs? As already mentioned, recognizing general (not necessarily optimal) nearly-planar graphs is often NP-hard. Restricting to bipartite graphs might allow for efficient recognition algorithms in some cases.

\item The maximal size of a complete graph that can be realized in the various models for nearly-planar graphs has been considered. For bipartite graphs, only weak bounds are known; see, e.g.,~\cite{Didimo2013}. Our results imply improved negative certificates. As we discuss in Section~\ref{sec:fanplanar}, we were able to realize $K_{5,5}-e$ as a fan-planar drawing and we conjecture that $K_{5,5}$ itself cannot be realized this way. Note that this would follow from a general upper bound of $4n-16$ edges for $n$-vertex bipartite fan-planar graphs.
For $2$-planar graphs, a direct application of the best previously known density bound only implies that $K_{5,8}$ cannot be realized. We conjecture that already $K_{5,5}$ is not $2$-planar.

\item Finally, one should study properties that not only hold for general nearly-planar graphs but also for bipartite ones, e.g., is every optimal bipartite RAC graph also~$1$-planar?
\end{enumerate}

\begin{figure}[t]
   \centering
   \includegraphics[scale=0.5,page=1]{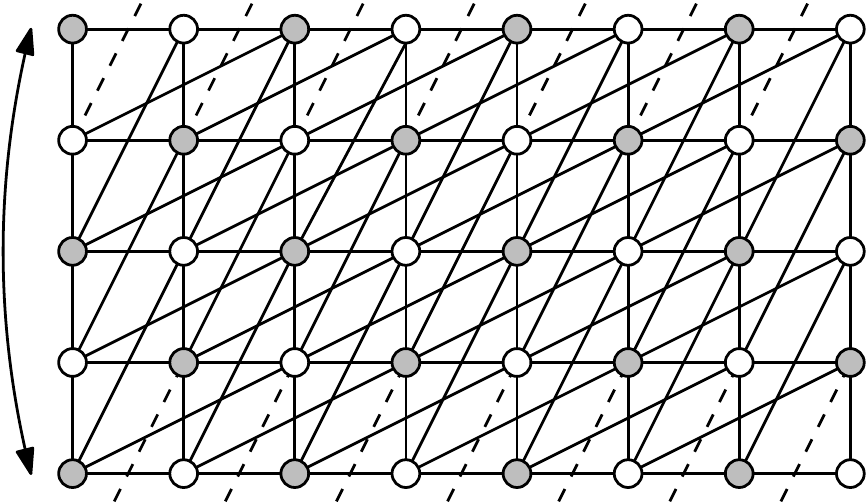}
   \caption{A bipartite $n$-vertex $3$-planar graph with $4n-O(1)$ edges; to see this observe that all vertices have degree~$8$, except for few vertices on the left and on the right which have smaller degree.}
   \label{fig:max-3-planar}
\end{figure}

\bibliographystyle{abbrv}
\bibliography{abbrv,references}
\end{document}